\documentclass[sigplan,screen]{acmart}

\usepackage[utf8]{inputenc} 
\usepackage[T1]{fontenc} 
\usepackage{microtype} 
\usepackage{fancyhdr}  
\usepackage{color}
\usepackage{xcolor}
\usepackage{etoolbox}

\usepackage{algorithm} 

\usepackage{mathtools}
\usepackage{mathptmx}

\usepackage{amsmath}
\usepackage{amsfonts}

\usepackage{bm}
\usepackage{enumerate}
\usepackage{lscape}
\usepackage{rotating}
\usepackage{xspace}
\usepackage{relsize}

\usepackage{multirow}
\usepackage{multicol}
\usepackage{booktabs} 
\usepackage{longtable}
\usepackage{array}
\usepackage{colortbl}

\usepackage{url}
\usepackage{textcomp}

\usepackage{wrapfig}
\usepackage[framemethod=TikZ]{mdframed}
\usepackage{tikz}
\usetikzlibrary{arrows,shapes,decorations,automata,backgrounds,calc,topaths,positioning,matrix}
\usepackage{lipsum}
\usepackage{balance}
\usepackage{listings}
\usepackage{graphicx}
\usepackage[normalem]{ulem}

\usepackage{subcaption}
\usepackage{hyperref}
\hypersetup{urlcolor=black,colorlinks=false,hidelinks,linkcolor=blue,filecolor=magenta,urlcolor=cyan}
\usepackage{cleveref}
\usepackage{natbib}

\usepackage[most]{tcolorbox}
\usepackage{listings}
\usepackage{textcomp}
\usepackage{xspace}

\newcommand{\originalgrumbler}[2]{\begin{quote}\textcolor{blue}{\sl{\bf #1 says:} #2}\end{quote}}
\newcommand{\grumbler}[2]{\originalgrumbler{#1}{#2}}
\newcommand{\swarnendu}[1]{\grumbler{Swarnendu}{#1}}

\newcommand{\anupam}[1]{\grumbler{Anupam}{#1}}

\newcommand{\later}[1]{\begin{quote}\textcolor{darkgreen}{\textbackslash \textbf{later\{}} #1 \textcolor{darkgreen}{\}}\end{quote}}

\newcommand{\savespace}[1]{\textcolor{purple}{SAVESPACE: #1}}

\definecolor{darkgreen}{rgb}{0,0.4,0}
\definecolor{mygreen}{rgb}{0,0.6,0}
\definecolor{mygray}{rgb}{0.5,0.5,0.5}
\definecolor{mymauve}{rgb}{0.58,0,0.82}
\definecolor{backgroundColour}{rgb}{0.95,0.95,0.92}

\lstdefinestyle{GenericStyle}{
    frame=tb, 
    backgroundcolor=\color{backgroundColour},
    basicstyle=\small,
    keywordstyle=\color{black}\textbf,
    stringstyle=\ttfamily,
    commentstyle=\color{darkgreen},
    showlines=true, 
    numbers=left,
    numberstyle=\tiny,
    stepnumber=1,
    numbersep=5pt,
    tabsize=2,
    showstringspaces=false
}

\lstdefinestyle{BashStyle}{
    language=bash,
    frame=tb, 
    backgroundcolor=\color{backgroundColour},
    basicstyle=\small\sffamily,
    linewidth=1.0\linewidth,
    xleftmargin=0.1\linewidth,
    showlines=true, 
    numbers=left,
    numberstyle=\tiny,
    stepnumber=1,
    numbersep=5pt,
    columns=fullflexible
}

\lstdefinestyle{JavaStyle}{
    language=Java, 
    frame=tb,
    backgroundcolor=\color{backgroundColour},
    basicstyle=\footnotesize\sffamily,
    keywordstyle=\color{blue}\textbf,
    commentstyle=\color{darkgreen}\textit,
    numbers=left,
    numberstyle=\tiny,
    stepnumber=1,
    numbersep=5pt,
    tabsize=2,
    extendedchars=true,
    breaklines=true,
    escapeinside={\%*}{*)}, 
    columns=flexible
}

\lstdefinestyle{CPPStyle}{
    language=C++,
    frame=tb, 
    backgroundcolor=\color{backgroundColour},
    basicstyle=\footnotesize\ttfamily,
    keywordstyle=\color{blue}\textbf,
    stringstyle=\color{red}\ttfamily,
    commentstyle=\color{darkgreen}\textit,
    morecomment=[l][\color{magenta}]{\#},
    literate = *{\ \ }{\ }1, 
    showlines=true, 
    numbers=left,
    numberstyle=\tiny,
    stepnumber=1,
    numbersep=5pt,
    tabsize=2
}

\lstdefinestyle{CStyle}{
    language=C,
    frame=tb, 
    backgroundcolor=\color{backgroundColour},
    basicstyle=\footnotesize\ttfamily,
    stringstyle=\color{mymauve},
    commentstyle=\color{mygreen}\textit,
    keywordstyle=\color{magenta},
    numberstyle=\scriptsize\color{mygray},
    breaklines=true,
    showlines=true, 
    numbers=left,
    numbersep=5pt,
    stepnumber=1,
    tabsize=2,
    showstringspaces=false
}

\lstdefinestyle{PythonStyle}{
    language=Python,
    frame=tb, 
    backgroundcolor=\color{backgroundColour},
    basicstyle=\ttm,
    keywordstyle=\ttb\color{deepblue},
    stringstyle=\color{deepgreen},
    otherkeywords={self},             
    emph={MyClass,__init__},          
    emphstyle=\ttb\color{deepred},    
    showstringspaces=false
}

\lstdefinestyle{CUDAStyle}{
    language=C++,
    frame=tb, 
    backgroundcolor=\color{backgroundColour},
    basicstyle=\footnotesize\ttfamily,
    keywordstyle=\color{blue}\textbf,
    stringstyle=\color{red}\ttfamily,
    commentstyle=\color{darkgreen}\textit,
    emph={
            cudaMalloc, cudaFree, __global__, __shared__, __device__, __host__, __syncthreads,
        },
    emphstyle=\color{darkgreen}\bf\ttfamily,
    numberstyle=\tiny\color{mygray},
    breaklines=true,
    showlines=true, 
    moredelim=[s][\ttfamily]{<<<}{>>>},
    numbers=left,
    numbersep=3pt,
    stepnumber=1,
    tabsize=2,
    showstringspaces=false
}

\newcommand{\bench}[1]{\textsf{\small #1}}
\newcommand{\code}[1]{\texttt{\small #1}}
\newcommand{\pcode}[1]{\textsf{#1}}

\newcommand{\eg}{e.g.\xspace}
\newcommand{\ie}{i.e.\xspace}

\newcommand{\etal}{et al.\xspace}

\newcommand{\naive}{na\"{\i}ve\xspace}

\newcommand{\Naively}{Na\"{\i}vely\xspace}

\tcbset{
    frame code={}
    center title,
    left=0pt,
    right=0pt,
    top=0pt,
    bottom=0pt,
    colback=gray!20,
    colframe=white,
    width=\dimexpr\textwidth\relax,
    enlarge left by=0mm,
    boxsep=5pt,
    arc=0pt,outer arc=0pt,
}

\def\CPP{{C\nolinebreak[4]\hspace{-.05em}\raisebox{.1ex}{\footnotesize\bf ++\xspace}}}

\graphicspath{{./figs/}}

\newcommand{\bigO}[1]{\( \mathcal{O} \left( #1 \right) \)}

\newcommand{\fastpar}{{Fast\-Racer}\xspace}
\newcommand{\Fastpar}{\fastpar}
\newcommand{\FastPar}{\fastpar}
\newcommand{\FASTPAR}{\fastpar}

\newcommand{\nft}{{Fast\-Racer}\xspace}

\newcommand{\spd}{{SPD3}\xspace}
\newcommand{\SPD}{\spd}

\newcommand{\ptracer}{{PTRacer}\xspace}
\newcommand{\PTRACER}{\ptracer}

\newcommand{\ft}{{Fast\-Track}\xspace}

\newcommand{\FT}{\ft}

\newcommand{\cracer}{{C-RACER}\xspace}
\newcommand{\CRACER}{\cracer}

\newcommand{\wsporder}{{WSP-Order}\xspace}
\newcommand{\WSPORDER}{\wsporder}

\newcommand{\thr}[1]{\textsf{\small #1}}
\newcommand{\task}[1]{\textsf{\small #1}}

\newcommand{\scalarclk}[1]{\textsf{#1}}
\newcommand{\vectorclk}[1]{\textbf{#1}}
\newcommand{\vecentry}[3]{\textsf{\small #1}\textsubscript{#2}(#3)}
\newcommand{\epoch}[2]{\textsf{\small #1@#2}}

\newcommand{\lca}[1]{\textsf{LCA}(#1)}

\newcommand{\ftsameepoch}{\textsc{Same Epoch}\xspace}

\newcommand{\ftexclusive}{\textsc{Exclusive}\xspace}
\newcommand{\ftshared}{\textsc{Shared}\xspace}
\newcommand{\ftshare}{\textsc{Share}\xspace}

\newcommand{\rdshared}{\textsc{Read Shared}\xspace}
\newcommand{\rdshare}{\textsc{Read Share}\xspace}

\newcommand{\blackscholes}{\bench{\small blackscholes}\xspace}
\newcommand{\bodytrack}{\bench{\small bodytrack}\xspace}
\newcommand{\fluidanimate}{\bench{\small fluidanimate}\xspace}
\newcommand{\fluidanimater}{\bench{\small fluidanimate-r}\xspace}
\newcommand{\streamcluster}{\bench{\small streamcluster}\xspace}
\newcommand{\streamclusterr}{\bench{\small streamcluster-r}\xspace}
\newcommand{\swaptions}{\bench{\small swaptions}\xspace}

\newcommand{\convexhull}{\bench{\small convexHull}\xspace}
\newcommand{\delrefine}{\bench{\small delRefine}\xspace}
\newcommand{\deltriang}{\bench{\small delTriang}\xspace}
\newcommand{\nn}{\bench{\bench{\small nearestNeigh}}\xspace}
\newcommand{\raycast}{\bench{\small rayCast}\xspace}

\newcommand{\kmeans}{\bench{\small kmeans}\xspace}

\newcommand{\karatsuba}{\bench{\small karatsuba}\xspace}
\newcommand{\sort}{\bench{\small sort}\xspace}

\newcommand{\ptracernotoolperfoverheadhimalaya}{9.44X\xspace}
\newcommand{\fastparnotoolperfoverheadhimalaya}{6.85X\xspace}
\newcommand{\cracernotoolperfoverheadhimalaya}{11.60X\xspace}

\newcommand{\fastparptracerspeeduphimalaya}{1.38X\xspace}
\newcommand{\fastparcracerspeeduphimalaya}{1.69X\xspace}

\newcommand{\ptracernotoolperfoverheadaravalli}{7.70X\xspace}
\newcommand{\ftnotoolperfoverheadaravalli}{50.26X\xspace}
\newcommand{\fastparnotoolperfoverheadaravalli}{6.41X\xspace}
\newcommand{\cracernotoolperfoverheadaravalli}{10.66X\xspace}

\newcommand{\fastparptracerspeeduparavalli}{1.20X\xspace}
\newcommand{\fastparcracerspeeduparavalli}{1.66X\xspace}
\newcommand{\ftfastparspeeduparavalli}{4.62X\xspace}

\renewcommand{\pcode}[1]{\textsf{\small #1}}

\newcommand{\asyncfinish}{async-finish\xspace}
\newcommand{\AsyncFinish}{Async-Finish\xspace}
\newcommand{\spawnsync}{spawn-sync\xspace}
\newcommand{\forkjoin}{fork-join\xspace}

\newtoggle{techreport}
\togglefalse{techreport}

\newtoggle{acmFormat}
\toggletrue{acmFormat}

\newtoggle{ieeeFormat}
\togglefalse{ieeeFormat}

\newtoggle{lncsFormat}
\togglefalse{lncsFormat}

\renewcommand{\grumbler}[2]{}
\renewcommand{\later}[1]{}
\renewcommand{\savespace}[1]{}

\iftoggle{lncsFormat}{
    \usepackage[T1]{fontenc}
    \def\doi#1{\href{https://doi.org/\detokenize{#1}}{\url{https://doi.org/\detokenize{#1}}}}

}{}

\iftoggle{ieeeFormat}{
\usepackage{cite}
\usepackage{amssymb}

\usepackage{algorithm}
\usepackage[noend]{algpseudocode}

\IEEEoverridecommandlockouts

\def\BibTeX{{\rm B\kern-.05em{\sc i\kern-.025em b}\kern-.08em
T\kern-.1667em\lower.7ex\hbox{E}\kern-.125emX}}

\newtheorem{lemma}{Lemma}
}{}

\iftoggle{acmFormat}{
    \usepackage{algorithmicx}
    \usepackage[noend]{algpseudocode}
    \algnewcommand\And{\textbf{and}\xspace}
    \algnewcommand\Or{\textbf{or}\xspace}

    \copyrightyear{2022}
    \acmYear{2022}
    \setcopyright{acmcopyright}\acmConference[PMAM'22]{The 13th International Workshop on Programming Models and Applications for Multicores and Manycores}{April 2--6, 2022}{Seoul, Republic of Korea}
    \acmBooktitle{The 13th International Workshop on Programming Models and Applications for Multicores and Manycores (PMAM'22), April 2--6, 2022, Seoul, Republic of Korea}
    \acmPrice{15.00}
    \acmDOI{10.1145/3528425.3529101}
    \acmISBN{978-1-4503-9339-3/22/04}

}{}

\begin{document}

\iftoggle{ieeeFormat}{
    \title{
        Efficient Data Race Detection of Task-Based Programs Using Vector Clocks
        \thanks{\IEEEauthorrefmark{1}{Both authors contributed equally to this research.}}
        \thanks{This material is based upon work supported by IITK Initiation Grant and SERB SRG/2019/000384.}

    }}{}

\iftoggle{acmFormat}{
    \title{
        Efficient Data Race Detection of \AsyncFinish Programs Using Vector Clocks
    }}{}

\iftoggle{acmFormat}{
    \author{Shivam Kumar}
    \email{shivamkm07@gmail.com}
    \affiliation{%
        \institution{Indian Institute of Technology Kanpur}
        \country{India}
    }
    \author{Anupam Agrawal}
    \email{anupamag@cse.iitk.ac.in}
    \affiliation{%
        \institution{Indian Institute of Technology Kanpur}
        \country{India}
    }
    \author{Swarnendu Biswas}
    \email{swarnendu@cse.iitk.ac.in}
    \affiliation{%
        \institution{Indian Institute of Technology Kanpur}
        \country{India}
    }

    \renewcommand{\shortauthors}{Kumar \etal}

    \begin{abstract}

    Existing data race detectors for task-based programs incur significant run time and space overheads. The overheads arise because of frequent lookups in fine-grained tree data structures to check whether two accesses can happen in parallel.

    \savespace{Happens-before tracking using vector clocks is a popular technique for sound and precise dynamic data race detection of multithreaded programs. However, a \naive application of vector clocks for task-based programs is inefficient because of the overhead to track accesses from many concurrent tasks.}
    This work shows how to efficiently apply vector clocks for dynamic data race detection of \asyncfinish programs with locks. Our proposed technique, \emph{\fastpar}, builds on the 
    \ft algorithm with per-task and per-variable optimizations to reduce the size of vector clocks. 
    \FastPar exploits the structured parallelism of \asyncfinish programs to use a coarse-grained encoding of the dynamic task inheritance relations
    to limit the metadata
    in the presence of many concurrent readers.
    Our evaluation shows that \fastpar substantially improves time and space overheads over \ft, and is competitive with the state-of-the-art data race detectors for \asyncfinish programs with locks.

\end{abstract}


}{}

\iftoggle{ieeeFormat}{
    \author{\IEEEauthorblockN{Anupam Agrawal\IEEEauthorrefmark{1}, Shivam Kumar\IEEEauthorrefmark{1}, and Swarnendu Biswas}
        \IEEEauthorblockA{\textit{Department of Computer Science and Engineering} \\
            \textit{Indian Institute of Technology Kanpur, India}\\
            \{anupamag,shivamkm,swarnendu\}@cse.iitk.ac.in}}
}{}

\maketitle

\iftoggle{ieeeFormat}{

}{}

\section{Introduction}
\label{sec:intro}


The task-based programming abstraction helps
write efficient and portable parallel code without having to think of low-level threads.
Tasks execute in parallel as hardware-agnostic logical units of work, and programmers only specify the dependencies among the tasks.
An accompanying runtime schedules tasks to threads and provides performance features like work-stealing to load-balance the execution.
Cilk~\cite{cilk,cilk-5}, X10~\cite{x10-oopsla-2005}, Habanero-Java~\cite{habanero-java},
\savespace{Intel Threading Building Blocks (TBB)~\cite{tbb-reinders},} and Java Fork-Join~\cite{java-fork-join} are popular task-based frameworks.

Task-based programs are susceptible to
concurrency errors such as atomicity violations~\cite{yoga-cgo-2016} and data races~\cite{spd3-pldi-2012,spbags,esp-bags-2010,cracer-spaa-2016,sp-order-spaa-2004}.
A \emph{data race} \savespace{in a task-based program} occurs when two accesses, with at least one write, from different tasks are incorrectly synchronized.
The presence of data races in shared-memory programs often indicates the presence of other concurrency errors~\cite{atomizer-2004},
and can affect an execution by crashing, hanging, or
corrupting data~\cite{portend-asplos12}. 
Data races are hard to detect and fix since they may occur nondeterministically
under specific thread interleavings, program inputs, and execution environments.
Data races have led to several real-world disasters~\cite{therac-25,blackout-2003,nasdaq-facebook};
such high-profile failures are a testament that
data races are present even in well-tested code.

\paragraph*{The problem}

There exists prior work to detect data races in task-based programs~\cite{ptracer-fse-2016,spd3-pldi-2012,cilk-spaa-98,spbags,esp-bags-2010,mellor-crumney-sc-1991,gtrace-europar-2018,sp-order-spaa-2004,cracer-spaa-2016}. Most
analysis utilize the series-parallel structure of execution of task-based programs to
check whether accesses can potentially execute in parallel (called \emph{series-parallel maintenance})~\cite{spbags,spd3-pldi-2012,ptracer-fse-2016,sp-order-spaa-2004,cracer-spaa-2016}.
\savespace{A good series-parallel (SP) maintenance data structure determines whether the race detection analysis runs in parallel.}
Prior techniques are either serial (\eg, \cite{spbags,esp-bags-2010,cilk-spaa-98}) or are difficult to parallelize (\eg, ~\cite{sp-order-spaa-2004}), detect races only in a given schedule (\eg, \cite{fasttrack}), continue to incur high run-time overheads (\eg, \cite{ptracer-fse-2016,spd3-pldi-2012,mellor-crumney-sc-1991}),
require tight coupling with the runtime scheduler for good performance (\eg, \cite{cracer-spaa-2016}), or do not support
lock-based synchronization (\eg, ~\cite{cracer-spaa-2016,spd3-pldi-2012,sp-order-spaa-2004}).

\paragraph*{Our approach}

In this work, we focus on efficient detection of \emph{per-input apparent data races in task-based programs with \asyncfinish semantics}.\footnote{Apparent data races occur because of the usage of parallel task constructs and ignore the per-schedule dynamic interleavings~\cite{race-conditions}. Feasible data races consider the nondeterministic timing variations during execution.}
For a given application and an input, per-input races include races observed in the current schedule as well as other schedules with possibly permuted memory operations, ignoring schedule-sensitive branches~\cite{ptracer-fse-2016}.
While prior work has \emph{ignored} how to optimize vector clocks for efficient race detection of task-based programs,
we \emph{argue} that analyses based on vector clocks are generic, inherently parallel, and have better data
locality than tree-based data structures for series-parallel (SP) maintenance.
\savespace{\FT~\cite{fasttrack,fasttrack2} is a
    popular data race detector for multithreaded programs that tracks the happens-before
    relation~\cite{happens-before} using vector clocks.}
However, a \naive application of vector-clock-based analysis to tasks results in
prohibitive memory and run-time overhead~\cite{ptracer-fse-2016,spd3-pldi-2012}.

This paper presents \emph{\fastpar}, a vector-clock-based data race detector for \asyncfinish programs. \Fastpar avoids the redundancy in \emph{per-task} vector clocks
by using auxiliary immutable data structures to maintain space- and time-efficient lossless vector clock representations correctly.
\fastpar exploits the structured parallelism in \asyncfinish programs
to optimize the space requirement of \emph{per-variable} metadata
in the presence of many concurrent readers.
Prior work has shown that a careful selection of only two ``concurrent read'' accesses
is sufficient for detecting read-write data races for \asyncfinish programs~\cite{spd3-pldi-2012,ptracer-fse-2016}.
\fastpar uses coarse-grained encoding of dynamic task inheritance relationships
to identify the two accesses (for both reads and writes) necessary for
race detection, and uses vector clocks to check whether two accesses can race.

We evaluate the performance and correctness
of \fastpar on 
\CPP applications that use Intel TBB~\cite{tbb-reinders} for task parallelism and compare with prior work~\cite{fasttrack,ptracer-fse-2016,cracer-spaa-2016}.
Our evaluation shows that the run time and memory overhead of \FastPar
is substantially lower compared to state-of-the-art data race detectors that target \asyncfinish programs.

\paragraph*{Contributions}


This paper makes the following contributions:

\begin{itemize}
    \item To the best of our knowledge, this work is the \emph{first} to show the viability of using vector clocks for efficient dynamic analysis of task-based programs;

    \item a race detector called \fastpar that detects per-input apparent races
          in \asyncfinish programs with locks,

    \item publicly available implementations of \fastpar and related techniques, and an evaluation that shows \fastpar significantly outperforms prior work.
\end{itemize}



\section{Background and Motivation\label{sec:background}}

This section briefly reviews data race detection of multithreaded programs using vector clocks. We also discuss closely related prior work on race detection of \asyncfinish programs.






\subsection{Race Detection with Vector Clocks\label{sec:vector-clocks}}

\savespace{Detecting data races by tracking the happens-before relation using vector clocks has several advantages. It is easy to encode the semantics of different programming models and synchronization mechanisms in the source language with vector clocks~\cite{barracuda-pldi-2017,scord-isca-2020,archer,google-tsan-v2}.}
Many race
detectors
for multithreaded programs use vector clocks to track happens-before (HB) relations
in an execution~\cite{multirace,pacer-pldi-2010,fasttrack,google-tsan-v2}.
Each thread maintains a scalar clock that is incremented at \emph{synchronization release} operations (\eg, lock release, monitor wait, thread fork and join, and volatile write). Each thread \thr{T} also maintains a vector clock \vectorclk{C} of size $n$, where there are $n$ threads in the application. The clock entry \vecentry{C}{\thr{T}}{\thr{U}} records the clock of thread \thr{U} when thread \thr{T} last synchronized with \thr{U}. A dynamic analysis updates per-variable vector clock metadata whenever a thread accesses a shared data or a lock variable.
Vector clock operations require \bigO{n} time and storage to monitor an
execution with $n$ threads.


\iftoggle{ieeeFormat}{
    \begin{small}
    \begin{algorithm}[ht]
        \caption{\small \FT analysis at synchronization operations}
        \label{alg:ft-sync-ops}
        \begin{algorithmic}[1]
            \Procedure{Spawn}{}
            \Comment{Thread \thr{T} spawns \thr{U}}
            \State \thr{U.vc} $\gets$ \thr{T.vc} $\cup$ \{\thr{T.epoch}\}

            \State \thr{T.epoch} $\gets$ \thr{T.epoch} + 1
            \Comment{Increment \thr{T}'s scalar clock}
            \State \thr{U.epoch} $\gets$ \thr{U.epoch} + 1
            \EndProcedure

            \Procedure{Join}{}
            \Comment{Thread \thr{T} joins with \thr{U}}
            \ForAll{<t,c> in \thr{U.vc}}
            \State \thr{T.vc[t]} $\gets$ \texttt{max}(\thr{U.vc[t]}, \thr{T.vc[t]})
            \EndFor
            \EndProcedure

            \Procedure{Acquire}{}
            \Comment{\thr{T} acquires lock \thr{L}}
            \ForAll{<t,c> in \pcode{L.vc}}
            \State \thr{T.vc[t]} $\gets$ \texttt{max}(\pcode{L.vc[t]}, \thr{T.vc[t]})
            \EndFor
            \EndProcedure

            \Procedure{Release}{}
            \Comment{\thr{T} releases lock \thr{L}}
            \State \pcode{L.vc} $\gets$ \thr{T.vc}
            \State \thr{T.epoch} $\gets$ \thr{T.epoch} + 1
            \EndProcedure

            \Function{CheckHB}{\epoch{c}{u},\thr{T}}
            \Comment{Check HB between epoch}
            \State \Return \epoch{c}{u} $\preceq$ \thr{T}.vc
            \EndFunction
        \end{algorithmic}
    \end{algorithm}
\end{small}

}{}

\iftoggle{acmFormat}{
    \begin{small}
    \begin{algorithm}[ht]
        \caption{\small \FT analysis at synchronization operations}
        \label{alg:ft-sync-ops}
        \begin{algorithmic}[1]
            \Procedure{Spawn}{}
            \Comment{Thread \thr{T} spawns \thr{U}}
            \State \thr{U.vc} $\gets$ \thr{T.vc} $\cup$ \{\thr{T.epoch}\}

            \State \thr{T.epoch} $\gets$ \thr{T.epoch} + 1
            \Comment{Increment \thr{T}'s scalar clock}
            \State \thr{U.epoch} $\gets$ \thr{U.epoch} + 1
            \EndProcedure

            \Procedure{Join}{}
            \Comment{Thread \thr{T} joins with \thr{U}}
            \ForAll{<t,c> in \thr{U.vc}}
            \State \thr{T.vc[t]} $\gets$ \texttt{max}(\thr{U.vc[t]}, \thr{T.vc[t]})
            \EndFor
            \EndProcedure

            \Procedure{Acquire}{}
            \Comment{\thr{T} acquires lock \thr{L}}
            \ForAll{<t,c> in \pcode{L.vc}}
            \State \thr{T.vc[t]} $\gets$ \texttt{max}(\pcode{L.vc[t]}, \thr{T.vc[t]})
            \EndFor
            \EndProcedure

            \Procedure{Release}{}
            \Comment{\thr{T} releases lock \thr{L}}
            \State \pcode{L.vc} $\gets$ \thr{T.vc}
            \State \thr{T.epoch} $\gets$ \thr{T.epoch} + 1
            \EndProcedure

            \Function{CheckHB}{\epoch{c}{u},\thr{T}}
            \Comment{Check HB between epoch}
            \State \Return \epoch{c}{u} $\preceq$ \thr{T}.vc
            \EndFunction
        \end{algorithmic}
    \end{algorithm}
\end{small}

}{}

\savespace{

\begin{table}
    \centering
    \footnotesize
    \begin{tabular}{llp{5.8cm}}
                                                                     & \textbf{State} & \textbf{Description}                                                               \\
        \toprule
        \multirow{5}{*}{\rotatebox[origin=c]{90}{\normalsize Read}}  & \ftsameepoch   & Thread reads a variable again before performing a synchronization operation        \\
                                                                     & \ftexclusive   & Current read happens after the previous read epoch                                 \\
                                                                     & \ftshare       & Current read is concurrent with previous read epoch                                \\
                                                                     & \ftshared      & Current read is concurrent with previous read vector                               \\
        \midrule
        \multirow{6}{*}{\rotatebox[origin=c]{90}{\normalsize Write}} & \ftsameepoch   & Thread writes a variable again before performing a synchronization operation       \\
                                                                     & \ftexclusive   & Read metadata is an epoch, current write happens after all previous accesses       \\
                                                                     & \ftshared      & Read metadata is a vector clock, current write happens after all previous accesses \\
    \end{tabular}
    \caption{Different analysis states in \ft~\cite{fasttrack}.}
    \label{tab:ft-states-new}
\end{table}

}

\savespace{
    \begin{figure}
        \centering
        \includegraphics[scale=0.45]{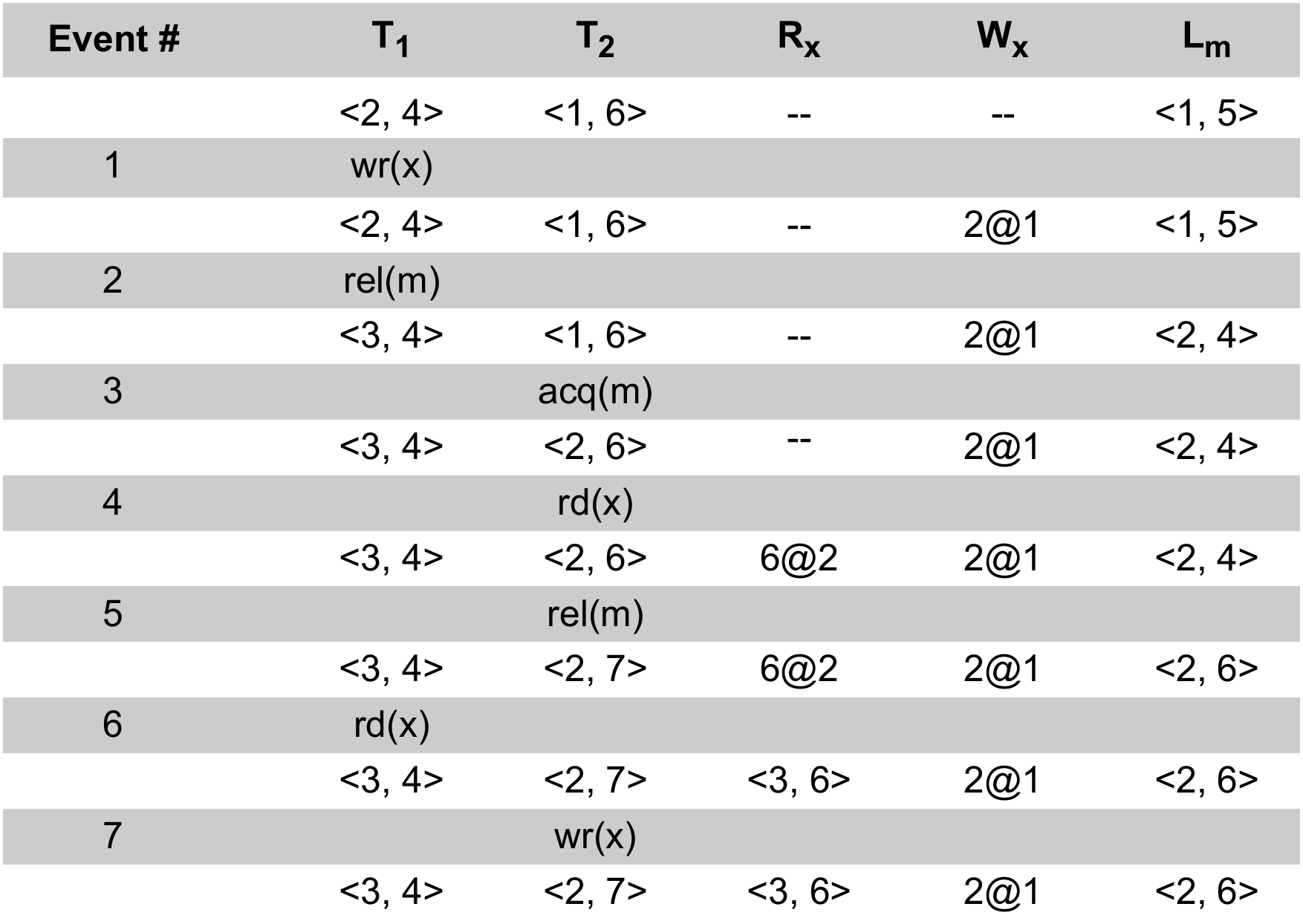}
        \caption{\ft algorithm for a program with threads \thr{T\textsubscript{1}} and \thr{T\textsubscript{2}}, lock variable \pcode{m}, and shared variable \pcode{x}.}
        \label{fig:ft-example}
    \end{figure}
}

\paragraph*{\FT}
\label{sec:ft}
The \emph{\FT} algorithm
tracks a single last writer and, in many cases, a single last reader~\cite{fasttrack}.
The total order on writes in a data-race-free program allows \FT to store only the \emph{last} write information. \FT stores the write metadata as an epoch \epoch{c}{T}, which is a tuple consisting of the writer thread identifier \thr{T} and the value of \thr{T}'s clock (say \scalarclk{c}) at the time of the write. The read metadata alternates between epoch and vector clock forms. An epoch representation suffices when there is a single reader or the current read happens after all previous reads. When there are concurrent readers, the read metadata is a vector clock (denoted by \pcode{vc}).
Algorithm~\ref{alg:ft-sync-ops} shows the pseudocode for the dynamic analysis performed by \ft at synchronization operations.
Before each access to a shared variable \pcode{x} by a thread \thr{T}, \FT checks whether the current access by \thr{T} happens after the previous write and all previous reads to \pcode{x} (\pcode{CheckHB},~\Cref{alg:ft-sync-ops}). A data race is reported if the current access by \thr{T} is concurrent with the last accesses. The shared data and lock variable metadata are updated upon a thread access.
\FT is popularly used as the basis for dynamically sound and precise\footnote{\emph{Sound} means no false negatives, and \emph{precise} means no false positives.} data race detection of multithreaded programs~\cite{barracuda-pldi-2017,pacer-pldi-2010,archer,google-tsan-v2,fib-oopsla-2017,slimfast}.



\savespace{
    Figure~\ref{fig:ft-example} shows an example interleaving and the corresponding metadata updates performed by \ft.
    Columns T\textsubscript{1} and T\textsubscript{2} show the operations performed by threads \thr{T\textsubscript{1}} and \thr{T\textsubscript{2}} and the thread vector clocks after each operation. Columns R\textsubscript{x} and W\textsubscript{x} represent the read and write metadata for the data variable \pcode{x}. L\textsubscript{m} shows the vector clock associated with lock variable \pcode{m}. The \ft algorithm detects a data race at the instruction number 8, since the vector clock of thread \thr{T\textsubscript{2}}
    does not happen after the last read  
    of variable \pcode{x} from \thr{T\textsubscript{1}}.
}

\savespace{
    However, \ft incurs substantial time ($\geq$5X) and space overheads~\cite{valor-oopsla-2015,fib-oopsla-2017,slimfast}. The high time overhead is due to the cost of \emph{synchronized} tracking of concurrent reads and for performing vector clock operations (can be \bigO{n}). The high space overhead is due to the need to maintain per-variable vector clocks for tracking concurrent readers and per-thread vector clocks to track inter-thread synchronization. Follow-up work optimizes the \ft analysis by removing redundant instrumentation~\cite{redcard-ecoop-2013}, reducing metadata redundancy~\cite{slimfast,slimstate}, and by limiting the analyses performed during execution~\cite{racemob,pacer-pldi-2010}.
}

\subsection{Race Detection for \AsyncFinish Programs}
\label{sec:async-finish}

\savespace{Structured task-parallel programming models enforce constraints over the unstructured model but still can express a wide range of parallel computations~\cite{spd3-pldi-2012}.
    Cilk~\cite{cilk,cilk-5}, X10~\cite{x10-oopsla-2005}, Habanero Java~\cite{habanero-java}, and Intel TBB~\cite{tbb-reinders} provide support for structured task parallelism.}

Frameworks like X10~\cite{x10-oopsla-2005} and Habanero Java~\cite{habanero-java} support  structured task parallelism with \asyncfinish semantics.
The statement ``\code{async \{\task{t}\}}'' creates a new child
task \task{t} that can run in series or in parallel with its parent task. The statement
``\code{finish \{\task{t}\}}'' causes the current task to wait for all the recursively-created tasks
within the block \task{t}. The \asyncfinish model is terminally-strict, which means each join edge
goes from the last instruction of a task to \emph{any} of its ancestors in the inheritance
tree~\cite{esp-bags-2010}.
\savespace{The \asyncfinish model is more general than the \spawnsync model,
    since every task in the \spawnsync model must join with its \emph{immediate} parent (\ie, fully-strict)~\cite{spd3-pldi-2012}.}
In the following, we discuss dynamic data race detection
techniques for 
\asyncfinish programs.

\begin{figure}[t]
    \centering
    \begin{subfigure}{\linewidth}
        \centering
        {\includegraphics[scale=0.29]{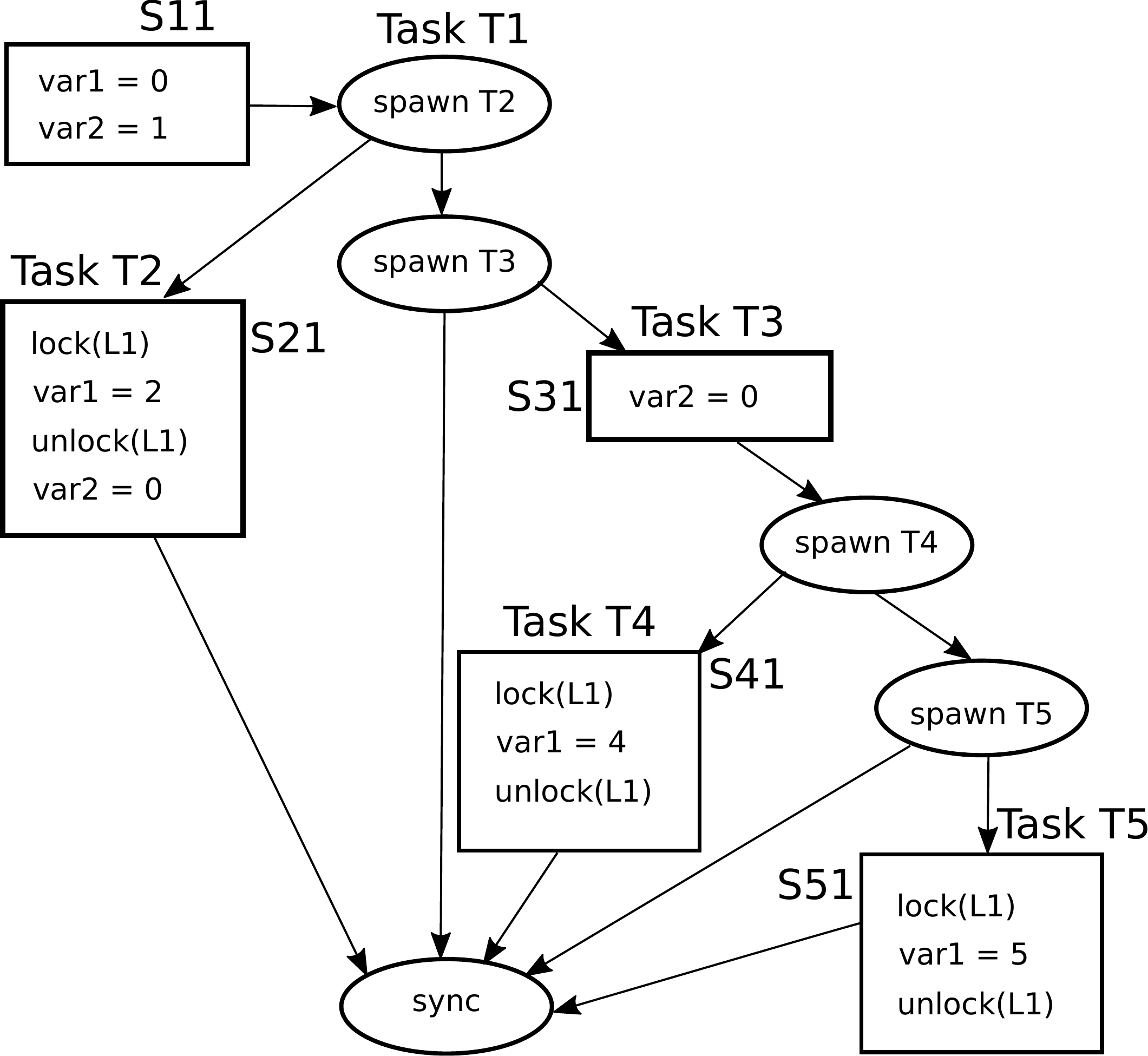}}
        \caption{An \asyncfinish program with tasks \task{T1}--\task{T5}, shared variables \code{var1} and \code{var2}, and a lock variable \code{L1}.}
        \label{fig:task-parallel-code}
    \end{subfigure}\\
    \begin{subfigure}{\linewidth}
        \centering
        {\includegraphics[scale=0.29]{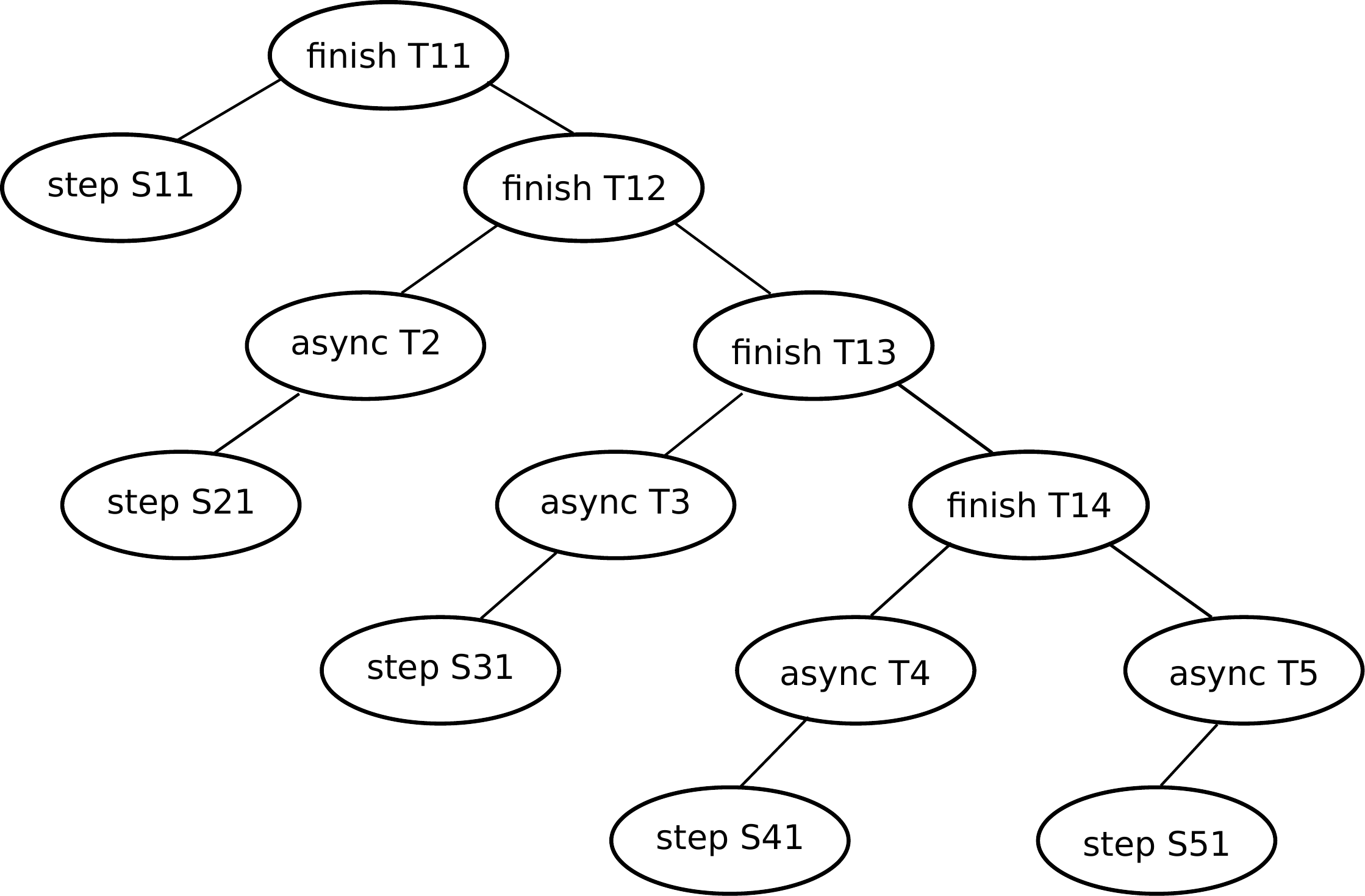}}
        \caption{The DPST for the program in Figure~\ref{fig:task-parallel-code}.
        }
        \label{fig:dpst}
    \end{subfigure}
    \caption{An \asyncfinish program and its DPST.}
    \label{fig:task-parallel-example}
\end{figure}

\paragraph*{\SPD}

\SPD~\cite{spd3-pldi-2012} uses a \emph{dynamic program structure tree} (DPST) to capture the
semantics of an \asyncfinish program.
A DPST consists of \pcode{step}, \pcode{finish}, and \pcode{async} nodes.
A \pcode{step} node represents the maximal sequence of instructions without any task management. An \pcode{async} node represents the spawning of a child task by a parent task. The descendants of an \pcode {async} node execute asynchronously with the remainder of the parent task. A \pcode{finish} node is created when a parent task spawns a child task and waits for the child, and its descendants, to complete. A \pcode{finish} node is thus the parent of all \pcode{async}, \pcode{finish}, and \pcode{step} nodes executed by its children or their descendants. Figure~\ref{fig:task-parallel-example} shows an \asyncfinish program and the corresponding DPST.
All executions of a data-race-free \asyncfinish program with the same input result in the same DPST~\cite{spd3-pldi-2012}.

The operational semantics of \asyncfinish programs imply a left-to-right computation order of sibling nodes belonging to a common parent task.
Thus, a DPST node's children are also ordered left-to-right to reflect the computation order
in their parent.
\savespace{\SPD exploits the structured parallelism 
    to optimize 
    SP maintenance.}
\savespace{\ie, check whether two memory accesses can happen in parallel.}
On a variable access, \spd searches for the \emph{lowest common ancestor} (LCA) of the current access 
(\ie, a \pcode{step} node)
and the last access stored in the variable's metadata. \SPD reports a race if the left child of
the LCA, which is an ancestor of the \pcode{step} node representing the last access, is an \pcode{async} node that indicates \emph{concurrent} execution of the last
access and the current task.
The DPST allows \spd to maintain one metadata
location for writes and two locations for reads in shadow memory.
\savespace{\spd \emph{ignores} analyzing locks in \asyncfinish programs.}


\savespace{\SPD is a sound and precise data race detector like \ft.}


\savespace{

\begin{figure*}
    \centering
    \includegraphics[width=\textwidth]{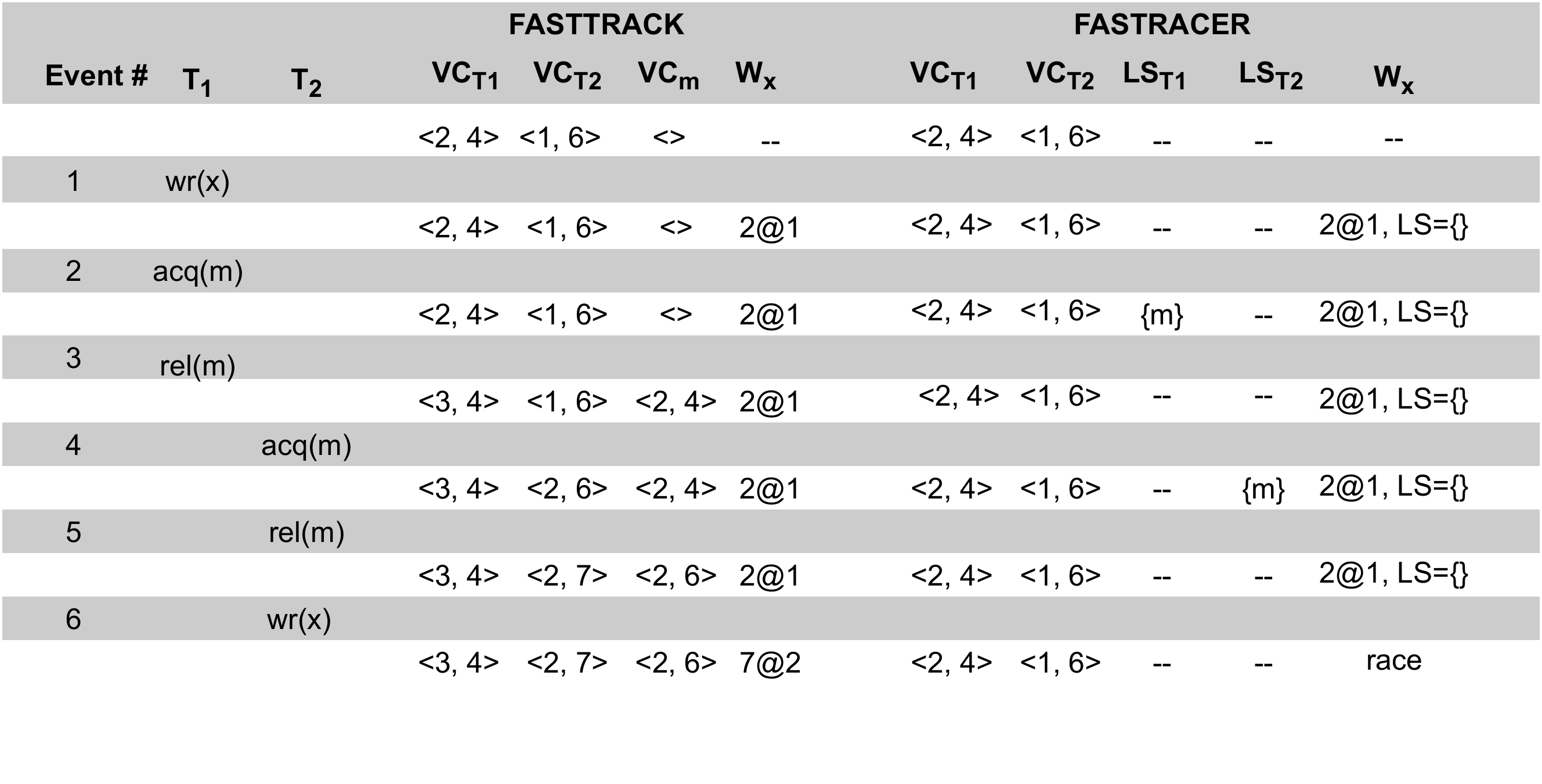}
    \caption{Comparison of \ft and \fastpar algorithm for a program with tasks \thr{T\textsubscript{1}} and \thr{T\textsubscript{2}}, lock variable \pcode{m} and shared variable \pcode{x}. \ft tracking happens-before relationships between lock acquire and release can miss races in some schedules, while \fastpar using LockSet(LS) technique to handle locks always detets races irrespective of schedule.}
    \label{fig:ft-fastpar-example}
\end{figure*}
}

\paragraph*{\ptracer}



\PTRACER extends \spd by detecting apparent races in \asyncfinish programs with locks~\cite{ptracer-fse-2016}.
\PTRACER 
maintains two metadata locations each for reads and writes to a shared variable.
The metadata per variable 
is proportional to the number of different locksets (\ie, set of locks held by the tasks at any time) with which the variable is accessed,
which is reasonable in practice because variables are usually accessed with similar locking
patterns.


\PTRACER selects two ``last read'' (``last write'') accesses from multiple parallel accesses with the same lockset to maintain constant metadata, such that any future write which can race with any one of the parallel reads (writes) will race with either one of the two chosen ``last readers'' (``last writers'').
\PTRACER makes these choices by selecting \pcode{step} nodes with the highest LCA among all parallel \pcode{step} nodes.
\PTRACER detects \emph{all} races for a given input even in the presence of lock-based synchronization.
Consider the shared variable \code{var1} which is updated by the parallel tasks \task{T2}, \task{T4}, and \task{T5} in Figure~\ref{fig:task-parallel-code}.
\savespace{The three parallel writes to \code{var1} do not race (protected by the lock \code{L1}).}
The step nodes corresponding to these accesses are \code{S21}, \code{S41}, and \code{S51}, respectively. Since, \code{S21} and \code{S51} have the highest LCA in the corresponding DPST,
\ptracer will store these two accesses and discard the access information for \code{S41}.

The race 
analysis in \ptracer is similar to \spd.
\PTRACER will report a race on \code{var2} for the example in Figure~\ref{fig:task-parallel-example}
because the left child of the LCA of step nodes \pcode{S21} and \pcode{S31} is an \pcode{async}
node.
\SPD will report false races on the variable \code{var1}.




\savespace{
    \subsection{Discussion}
    \label{sec:motivation}
}

\ptracer uses the DPST to maintain a constant amount of per-variable metadata independent of the
number of tasks executing the program. Furthermore, \ptracer performs frequent lookups in
the DPST to check whether two accesses can happen in parallel. However, the DPST can be deep for
programs with a recursive pattern of task creation and large because of many \pcode{step} nodes.
These lead to high run time and memory overheads.
\PTRACER uses an
array-based representation of the DPST and caches LCA lookups to improve the performance of LCA.
However, the DPST and the LCA computation continue to be a significant bottleneck for several
benchmarks, as we show in Section~\ref{sec:eval:perf}.
Thus, there is a need for \emph{more efficient} techniques to help detect data races in \asyncfinish programs.

\savespace{
    \medskip
    \noindent Table~\ref{tab:comparison} summarizes the 
    qualitative differences between \ft, 
    \wsporder, and \ptracer for data race detection of task-based programs. We discuss our proposed approaches \nft and \fastpar in the following sections.
}


\section{\fastpar: Efficient Dynamic Data Race Detection for \AsyncFinish Programs}
\label{sec:fastpar}

The thesis of this work is that vector clocks can provide better data locality for series-parallel (SP) maintenance in task-based programs compared to tree-based data structures used in prior work (\eg, \cite{cracer-spaa-2016,spd3-pldi-2012,ptracer-fse-2016}).
We present \emph{\fastpar}, a novel algorithm that integrates the benefits of vector-clock-based analysis \emph{and} exploits structured parallelism of \asyncfinish programs to limit the
amount of per-variable metadata.

Task-based programs create more, often orders of magnitude, parallel tasks than threads in multithreaded programs.
For example, the benchmark \nn
creates up to $\sim~3\times10^{6}$ tasks (Table~\ref{tab:run-time-stats}).
Race detectors like \ft use per-thread and per-variable vector clocks to capture the clock values of all the threads in the system.
While this representation works fine for multithreaded programs where the number of concurrent threads is comparatively small
(\(\sim\)\#cores), it is impractical for task-based programs and leads our
\ft implementation
to run out of memory on several benchmarks (Section~\ref{sec:eval:perf}).
Storing
only non-zero entries in a vector clock does not help since several concurrent tasks potentially access shared read-only variables.
Furthermore, maintaining vector clocks proportional to the number of threads can detect only feasible data races and misses races among concurrent tasks~\cite{romp-sc-2018}.
In the following, we discuss novel ideas
to solve these challenges.

\iftoggle{ieeeFormat}{
    \begin{small}
    \begin{algorithm}[htp]
        \caption{\small \fastpar analysis at synchronization operations}
        \label{alg:fastpar-sync-ops}
        \begin{algorithmic}[1]
            \Procedure{Spawn}{}
            \Comment{Task \task{T} spawns task \task{U}}

            \If{\texttt{size}(\task{T.rw\_vc}) $>$ \textsf{\smaller THRESHOLD}}
            \State \task{T.ro\_vc} $\gets$ REF(\{\task{T.ro\_vc} $\cup$ \task{T.rw\_vc}\})
            \State \task{T.rw\_vc} $\gets$ $\emptyset$
            \Else
            \State \task{U.ro\_vc} $\gets$ \task{T.ro\_vc}
            \State \task{U.rw\_vc} $\gets$ \task{T.rw\_vc}
            \EndIf

            \State \task{U.rw\_vc} $\gets$ \task{U.rw\_vc} $\cup$ \{\task{T.epoch}\}
            \State \task{U.joined} $\gets$ \task{T.joined}
            \State \task{U.lockset} $\gets$ \task{T.lockset}
            \State \task{U.IVC} $\gets$ \task{T.IVC} $\cup$ \{\texttt{getClock}(\task{U.epoch})\}
            \State \task{T.epoch} $\gets$ \task{T.epoch} + 1
            \State \task{U.epoch} $\gets$ \task{U.epoch} + 1

            \EndProcedure

            \Procedure{Join}{}
            \Comment{Task \task{T} joins with \task{U}}
            \State \task{T.joined} $\gets$ \task{T.joined} $\cup$ \{\texttt{getTaskId}(\task{U.epoch})\}
            \EndProcedure

            \Function{CheckHB}{\epoch{c}{u},\task{T}}
            \State
            \Comment{Check HB between epoch \epoch{c}{u} and
                \task{T}'s access}
            \State \Return \epoch{c}{u} $\preceq$ \task{T.rw\_vc} \textbf{or} c@u $\preceq$ \task{T.ro\_vc} \textbf{or} \task{u} $\in$ \task{T.joined}
            \EndFunction
        \end{algorithmic}
    \end{algorithm}
\end{small}

}{}

\iftoggle{acmFormat}{
    
}{}

\swarnendu{We should say somewhere that we are dealing with logical concurrency and so cannot deal with only those tasks mapped to threads.}

\subsection{Adapting Vector Clocks for Task Parallelism\label{sec:fastpar:nft-design}}

We analyzed the performance of \ft
and found that operations on task vector clocks incur high time and space overhead.
For example, the task join operation is a bottleneck because it requires
comparing and merging \emph{all} the clock values in the child and the parent tasks' vector clocks.
\Naively merging the vector clocks is not required since most vector clock entries remain unchanged during the lifetime of a task.

\paragraph*{Tracking read-only clock entries}

The first insight in \fastpar
is that \emph{most} vector clock entries for a task remain unchanged during the lifetime of the task, and the clock values continue to be the same as in the parent task. Thus, maintaining per-task copies of the clock entries is mostly redundant.

In \fastpar, a task vector clock is partitioned into a read-only (denoted by \pcode{ro\_vc}) and a read-write portion (denoted by \pcode{rw\_vc}). Child tasks in \fastpar maintain a reference to the
\pcode{ro\_vc} of their parents instead of maintaining redundant copies.
During a \pcode{spawn} operation (\savespace{\pcode{\smaller SPAWN}, }Algorithm~\ref{alg:fastpar-sync-ops}), \fastpar first checks if the size of \pcode{rw\_vc} is greater than a threshold. If yes, then
\fastpar merges \pcode{ro\_vc} and \pcode{rw\_vc} of the parent task into a new \pcode{ro\_vc} for the child task and \pcode{rw\_vc} of child task is kept empty.
Otherwise, the child \pcode{ro\_vc} points to the parent \pcode{ro\_vc} and the
parent's \pcode{rw\_vc} is copied to the child's \pcode{rw\_vc}.
Avoiding needless copies and redundant operations on the read-only portions of vector clocks helps
reduce space overheads \emph{and} improve performance. In the common case, most vector clock entries
of a task remain unchanged, \ie, the size of \pcode{rw\_vc} is small. Complete vector clock copies
happen only when the size of \pcode{rw\_vc} is greater than \pcode{\smaller THRESHOLD}.

\paragraph*{Optimizing vector clock join}

An access to a shared variable by a parent task after joining with a child task always happens after the child's accesses, since the accesses are synchronized by the \pcode{join} operation.
The second insight is that \fastpar
does not need to store the clock values of the child tasks \emph{after} the join operation. Instead, tracking the set of all child tasks that have joined with the parent task suffice.

Each task in \fastpar maintains the set of child tasks that have already joined with it in a \pcode{joined} data structure.
No vector clock join occurs when a task \task{T} joins with task \task{U},
instead, the child task is added to the \pcode{joined} of the parent task (\savespace{\pcode{\smaller JOIN}, }Algorithm~\ref{alg:fastpar-sync-ops}).
When a parent task spawns a new child task, the parent \code{joined} is copied to the child \code{joined}.
\savespace{Algorithm~\ref{alg:ft-sync-ops} shows that the complete vector clock is used for all the task operations in \ft, which is expensive and often unnecessary. Algorithm~\ref{alg:nft-sync-ops} shows the optimized spawn and join algorithms in \fastpar.}
The size of task vector clocks reduces significantly due to the \pcode{joined} optimization. Note that a trivial implementation of \pcode{joined} may not give performance benefits over \ft. We use efficient set implementations based on the Union-Find data structure, which reduces the overhead of set operations.

\paragraph*{Vector clock caching}

The vector clocks for a few tasks can be large even after the optimizations.
\FastPar uses explicit attributes to cache recently used vector clock values to avoid the cost of looking up the map data structure representing vector clocks. \FastPar indexes into the vector clock map when the thread id is not among the most recently used.

\subsection{Specializing for \AsyncFinish Programs}
\label{sec:fastpar:new-algo}


It can be expensive to 
maintain all concurrent readers 
of a shared variable in task-based programs.
\fastpar uses coarse-grained tracking of task inheritance relationships
to select relevant accesses from parallel readers/writers,\footnote{When tasks use locks,
    two writes can happen in parallel but do not constitute a race if they are protected by the same lock.}
which allows maintaining constant per-variable metadata per-lockset (i.e., the set of locks held by the task).


\begin{figure}
    \centering
    \includegraphics[scale=0.38]{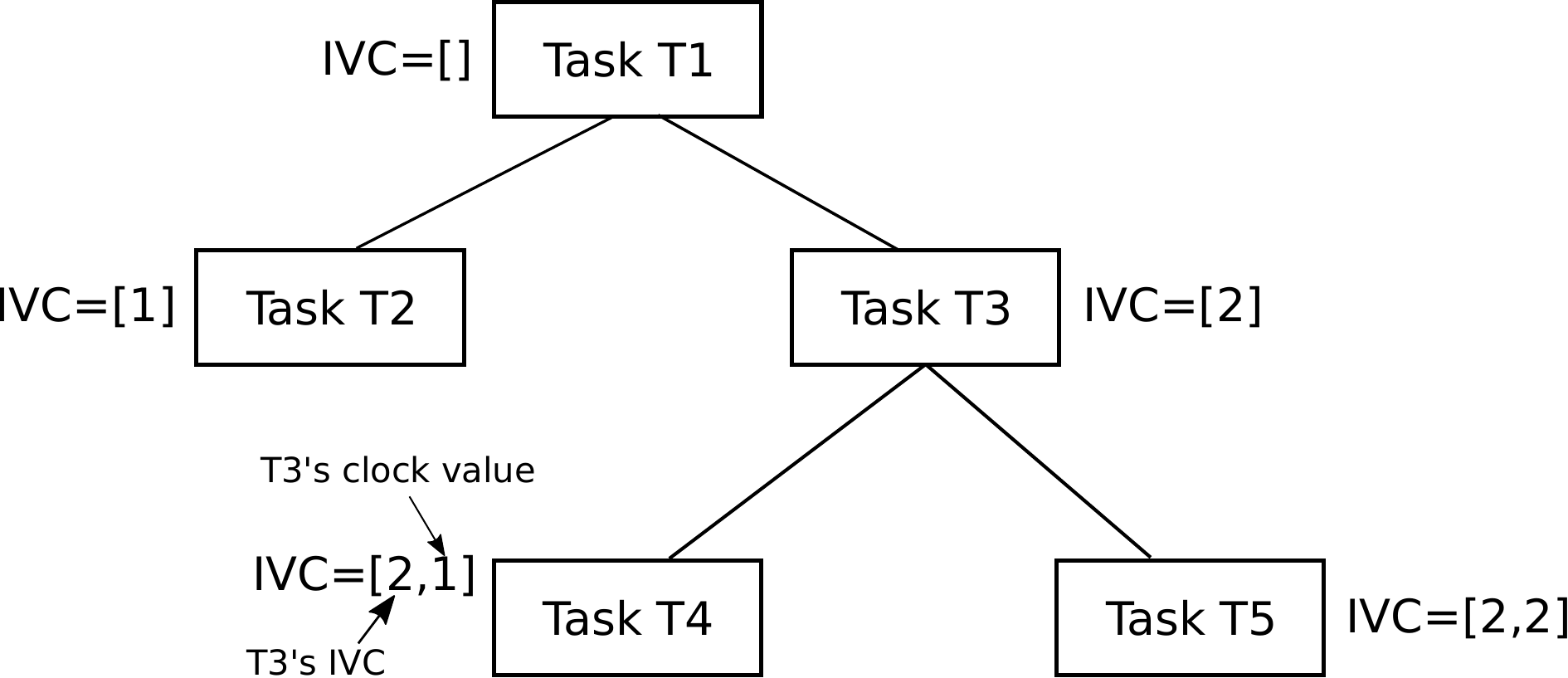}
    \caption{The inheritance tree 
        for program in Figure~\ref{fig:task-parallel-code}.}
    \label{fig:inheritance-graph}
\end{figure}

\paragraph*{Maintaining constant per-variable metadata}

\FastPar models the 
parent-child relationship among different tasks
with a \emph{task inheritance} tree. The nodes represent tasks, and edges represent the creation of child tasks by the parent.
Figure~\ref{fig:inheritance-graph} shows the inheritance tree for the 
program shown in Figure~\ref{fig:task-parallel-code} (ignore the IVC labels for now).
\FastPar uses the inheritance tree to \emph{efficiently} select two accesses out of multiple concurrent accesses with the same lockset, such that any racy future access that races with any one of the concurrent accesses must be racy with either one of the two chosen last accesses.
For parallel tasks accessing a variable with the same lockset, \fastpar stores the access
history of the two tasks with the highest LCA in the inheritance tree.
Consider the three accesses to \code{var1} from tasks \task{T2}, \task{T4}, and \task{T5} in \Cref{fig:task-parallel-code}.
In the inheritance tree shown in Figure~\ref{fig:inheritance-graph},
task nodes \task{T2} and \task{T4} (or \task{T5}) have the highest LCA, and so \fastpar stores the access
histories from \task{T2} and \task{T4} and discards \task{T5}. Any later access to \code{var1}, which is not parallel with both \task{T2} and \task{T4}, will not be parallel with \task{T5}.
So, discarding \task{T5}'s access information is correct.\footnote{In \pcode{\smaller async-finish} semantics, a task must join with its ancestor, either immediate or recursive. In case a parent task calls \pcode{\smaller join}, all children tasks within this join scope, either immediate or recursive, join with it.}
The metadata stored per shared variable in \fastpar is proportional to the number of the different locksets with which the variable has been accessed.

\savespace{Similar to \fastpar, \ptracer selects two accesses out of multiple concurrent accesses with the same lockset using a specific tree called DPST.}
The primary difference between a DPST 
and \fastpar's inheritance tree is in the \emph{granularity} of the nodes.
While a DPST decomposes a task into several unsynchronized regions represented by \pcode{step} nodes, an inheritance tree has just one node per task. The coarser modeling makes the inheritance tree much smaller and shallower than the DPST. Unlike a DPST, there is no left-to-right ordering in an inheritance tree.
While \ptracer uses DPST to check the concurrency between two accesses and select accesses with the highest LCA, \fastpar only does the latter with the inheritance tree. \fastpar uses vector clocks for data race detection to compensate for the coarser modeling and loss of ordering between the nodes (\Cref*{sec:fastpar:data-races}).

\paragraph*{Inheritance vector clock}



Instead of building an inheritance tree, \fastpar encodes the inheritance relations in a per-task array of clock values called \emph{Inheritance Vector Clock} (IVC)
for better performance.
An IVC is an immutable vector clock
that contains the clock values of all the 
reachable parents of a task \task{T} at the time of the creation of \task{T}. An IVC identifies the unique path in the inheritance tree from the root task to \task{T}, since a task spawn increments the scalar clock of the caller task. Whenever a parent task creates a child task, \fastpar copies the parent's IVC to the child and appends the parent's 
clock value at the end of child IVC. In Figure~\ref{fig:inheritance-graph}, the IVC of the parent task \task{T3} is copied to the child task \task{T4} and the current clock of \task{T3} (assumed to be one) is appended.

Both IVC labeling and the  Offset-Span (OS) labeling~\cite{mellor-crumney-sc-1991}
schemes compute a unique label from the label of the immediate predecessor, and guarantee that the length of a task's label will always be proportional to the
depth of the task in the inheritance tree.
However, there are two differences. 
First, the IVC of a task, once created, is immutable. Unlike OS labeling, any further task join operations do not modify the IVC.
Second, the Span part of OS labeling can only be assigned after the task graph is complete, \ie, OS labels cannot be computed while building the task graph.
Intuitively,
IVC labels are more similar to the Offset part,
with the constraint that they do not
get updated at \pcode{join} operations.
Whereas prior work uses OS labels for race detection~\cite{mellor-crumney-sc-1991,romp-sc-2018}, \fastpar uses IVC labels to identify two accesses out of multiple parallel accesses having the highest LCA in the inheritance tree.


\paragraph*{Computing tasks with the highest LCA}
\label{lca-computation}

To compute the highest LCA among three parallel tasks, \fastpar iterates over the IVCs of all the tasks simultaneously and stops at the first point of difference. After that, \fastpar chooses the task with a different clock value and any one of the other two. If all the three clock values are different, any two out of three tasks can be chosen. If \fastpar reaches the end of an IVC for any task \task{T}, it indicates that \task{T} must be the parent of the other two in the inheritance tree, and so, \fastpar chooses the parent task \task{T} and any one from the other two. The tasks so chosen will have the highest LCA in the inheritance tree (see \Cref*{lemma:two}).
The computation overhead depends on the sizes of the IVC, which is equal to the height of the inheritance tree. Since an inheritance tree is much shallower compared to a DPST, the computation is relatively efficient.


In 
Figure~\ref{fig:task-parallel-code}, the accesses to 
\code{var1} from three parallel tasks, \task{T2}, \task{T4}, and \task{T5}, are with the same lockset \code{\{L1\}}. The IVCs of the corresponding task nodes,
shown in Figure~\ref{fig:inheritance-graph}, are IVC$_{T2}$=[1], IVC$_{T4}$ = [2,1], and IVC$_{T5}$ = [2,2].
To select two out of the three accesses, \fastpar iterates over the three IVCs simultaneously. The first point of difference is at the first entry itself: IVC$_{T2}$[0]=1, IVC$_{T4}$[0]=IVC$_{T5}$[0]=2. So, \fastpar selects a task with a different clock value, \ie, \task{T2} and any one of other two, say \task{T4} (or \task{T5}). For the example in \Cref*{fig:inheritance-graph}, tasks \task{T2} and \task{T4} (or \task{T5}) indeed have the highest LCA in the inheritance tree.

\subsection{Detecting Data Races}
\label{sec:fastpar:data-races}

\FASTPAR extends the \ft algorithm (Section~\ref{sec:ft}) to detect races and update metadata.






\paragraph*{Metadata}

\begin{figure}
    \begin{subfigure}[t]{.48\textwidth}
        \caption{Metadata maintained in \ft.}
        \label{fig:ft-metadata}
        \centering
        \begin{lstlisting}[style=CPPStyle]
class PerVariableMetadata {
    epoch wr_md, rd_md; // Write and Read epoch
    map<taskid, clock> rd_vc; // Read vector clock
}
class PerLockMetadata {
    map<taskid, clock> vc;
}
class PerTaskMetadata {
    epoch epoch;
    map<taskid, clock> vc;
}\end{lstlisting}
    \end{subfigure}
    \hspace{5pt}
    \begin{subfigure}[t]{.48\textwidth}
        \caption{Metadata maintained in \fastpar.}
        \label{fig:fastpar-metadata}
        \centering
        \begin{lstlisting}[style=CPPStyle]
class PerVariableMetadata {
    PerLockMetadata lock[];
}
class PerLockMetadata {
    set<lockID> lockset;
    epoch rd1, rd2;    // Reader 1 and 2 metadata
    IVC* rd1_ivc, rd2_ivc; // IVC of Reader 1 and 2
    epoch wr1, wr2;    // Writer 1 and 2 metadata
    IVC* wr1_ivc,wr2_ivc; // IVC of Writer 1 and 2
}
class PerTaskMetadata {
    epoch epoch;
    map<taskid, clock> ro_vc; // Rd-only vector clock
    map<taskid, clock> rw_vc; // Rd-wr vector clock
    set<taskid> joined;  // All joined child tasks
    set<lockId> lockset; // Locks held by the task
    int IVC[];
}\end{lstlisting}
    \end{subfigure}
    \caption{Comparison of metadata state maintained in \ft and \fastpar.}
    \label{fig:metadata-states}
\end{figure}


Figure~\ref{fig:metadata-states} compares the metadata maintained in \ft and \fastpar. Given a task \task{T}, \ft uses a vector clock 
for storing the epoch values of \emph{all} other tasks (line 10, \Cref*{fig:ft-metadata}).
The attributes \pcode{ro\_vc}, \pcode{rw\_vc}, and \pcode{joined} in \code{PerTaskMetadata} in \Cref*{fig:fastpar-metadata} correspond to the optimizations introduced in
\Cref*{sec:fastpar:new-algo}.
While \ft uses vector clocks to track lock operations,
\fastpar uses the lockset mechanism~\cite{eraser}. Every task in \fastpar maintains
and updates \pcode{lockset} at lock acquire and release operations.

The per-variable state in \fastpar is an array of lock metadata. Each lock metadata stores
the access history of the variable with the distinct set of locks held before the
access. Every per-lock metadata contains four access history entries, two each for
reads and writes, which contain the epoch value and a reference to the IVC of the task at the time
of access.
Storing references to the IVCs suffice since they remain unchanged during the lifetime of a task.

\paragraph*{Race checks}

\savespace{

\begin{small}
  \begin{algorithm}[t]
    \caption{\hfill \fastpar: Task \task{T} accesses a shared variable \code{var}}
    \label{alg:fastpar-racecheck-algo}
    \begin{algorithmic}[1]
      \State \code{varMD} $\gets$ getMetadata(\code{var})
      \State \code{lockset} $\gets$ getLockSet(\code{varMD})
      \For{\code{lock} in \code{lockset}}
      \If{getLockSet(\code{lock}) $\cap$ getLockSet(\task{T}) == $\phi$}
      \If{isRead(\task{T}, \code{var}) \textbf{and} [\textbf{not} HB\_Check(\code{lock}.m\_wr1, \task{T}) \textbf{or} \textbf{not} HB\_Check(\code{lock}.m\_wr2, \task{T})]}
      \State Report Write--Read race 
      \EndIf
      \If{isWrite(\task{T}, \code{var}) \textbf{and} [\textbf{not} HB\_Check(\code{lock}.m\_rd1, \task{T}) \textbf{or} \textbf{not} HB\_Check(\code{lock}.m\_rd2, \task{T})]}
      \State Report Read--Write race 
      \EndIf
      \If{isWrite(\task{T}, \code{var}) \textbf{and} [\textbf{not} HB\_Check(\code{lock}.m\_wr1, \task{T}) \textbf{or} \textbf{not} HB\_Check(\code{lock}.m\_wr2, \task{T})]}
      \State Report Write--Write race 
      \EndIf
      \EndIf

      \If{getLockSet(\code{lock}) == getLockSet(\task{T}) }

      \If{isRead(\task{T}, \code{var})}
      \If{HB\_Check(\code{lock}.m\_rd1, \task{T}) or HB\_Check(\code{lock}.m\_rd2, \task{T})}
      \State Update corresponding reader history
      \EndIf
      \If{\task{T} in HighestLCA(\task{T}.IVC, \code{lock}.m\_rd1\_ivc, \code{lock}.m\_rd2\_ivc)}
      \State Replace any one reader with T
      \EndIf
      \EndIf

      \If{isWrite(\task{T}, \code{var})}
      \If{HB\_Check(\code{lock}.m\_wr1, \task{T}) or HB\_Check(\code{lock}.m\_wr2, \task{T})}
      \State Update corresponding writer history
      \EndIf
      \If{\task{T} in HighestLCA(\task{T}.IVC, \code{lock}.m\_wr1\_ivc, \code{lock}.m\_wr2\_ivc)}
      \State Replace any one writer with \task{T}
      \EndIf
      \EndIf
      \EndIf
      \EndFor

      \If{getLockSet(\task{T}) $\notin$ lock}
      \State \code{lock} $\gets$ \code{lock} $\cup$ createNewLockstate(\task{T})
      \EndIf
    \end{algorithmic}
  \end{algorithm}
\end{small}

}


When a shared variable is accessed, \fastpar iterates over all the lock metadata corresponding to distinct locksets with which the shared memory variable has been accessed. An empty intersection of the lockset of the current access and the lock metadata implies potentially parallel accesses. If the two locksets are disjoint, \FastPar checks if the epoch values stored in the access history happens before the current access using vector clocks (\pcode{\smaller CHECKHB}, \Cref*{alg:fastpar-sync-ops}).
If there is no such relationship, it implies that the prior access is concurrent with the current access. Finally, \fastpar reports a data race if one of the two accesses is a write.

Before accessing a shared variable \pcode{x}, \ft compares the current task's vector clock with the epoch(s) stored in \pcode{x}'s access history to determine if the current access 
to \pcode{x} happens \emph{after} the past accesses (\pcode{\smaller CHECKHB}, Algorithm~\ref{alg:ft-sync-ops}). \fastpar, apart from the vector clock entry check, also checks if the tasks present in \pcode{x}'s access history belong to \pcode{joined} of the current task. If all the tasks are present in \pcode{joined}, \fastpar infers that the current access 
happens \emph{after} prior accesses.
Since the vector clock is spread across 
\code{ro\_vc}, \code{rw\_vc}, and \code{joined}, 
\pcode{\smaller CHECKHB} (\Cref*{alg:fastpar-sync-ops}) checks the HB relation against all of them.

\paragraph*{Metadata updates}

\FASTPAR updates the read metadata corresponding to the current lockset if a read does not race with prior writes. \FastPar checks if any of the read epochs in the lock metadata corresponding to the current lockset happens before the current task's access.
If yes, then \FastPar updates that access entry with the current task's epoch and IVC. Otherwise, there are three parallel reads, and \fastpar needs to select two with the highest LCA. \FastPar iterates over the IVC of all three access entries and stops either at the first point of difference or if one of the IVCs end. \fastpar stores the access history of the task corresponding to the selected IVC and any one of the other two. Using any one of the other two works since, in both the cases, the two chosen tasks will have the highest LCA in the inheritance tree (Section~\ref{lca-computation}).
Figure~\ref{fig:tusker-example1} shows an example of how \fastpar updates the read metadata.
Assume tasks \task{T2}, \task{T3}, and \task{T4} all read a shared variable \code{x} and task
\task{T5} writes \code{x}. After the reads from \task{T2} and \task{T3}, the two readers stored for the variable \code{x} are \task{T2} and \task{T3}, because both these tasks can run in parallel.
Since task \task{T4} is spawned by \task{T1} after \task{T1} synchronizes with \task{T2}, so the
read by \task{T4} happens after the read of \task{T2}. The read metadata entry of \task{T2} is
replaced by \task{T4}. Next, when \task{T5} writes \code{x}, \fastpar checks the access with all
previous reads and reports a read-write race between the accesses from \task{T4} and \task{T5}.
The steps performed by \fastpar on a write access are similar.



\paragraph*{Synchronization operations}

During a \pcode{spawn} operation (Algorithm~\ref{alg:fastpar-sync-ops}), \fastpar checks if the size of \pcode{rw\_vc} is greater than a threshold. If yes, \fastpar merges \pcode{ro\_vc} and \pcode{rw\_vc} of the parent task into a new \pcode{ro\_vc} for the child task and \pcode{rw\_vc} of the child is kept empty. Otherwise, the child \pcode{ro\_vc} references the parent \pcode{ro\_vc}, and the child \pcode{rw\_vc} is copied from parent \pcode{rw\_vc}. Thereafter, \fastpar copies the parent's \pcode{joined} and \pcode{lockset}
to the child's \pcode{joined} and \pcode{lockset}, respectively.
In case of a \pcode{join} operation, the \pcode{joined} of the parent task is updated to contain the id of the child task.

A task's vector clock is not updated
in case of lock acquire and release operations. Instead, \fastpar uses \ptracer's mechanism to deal with lock operations. Each task maintains a \pcode{lockset}.
When a task \task{T} acquires a lock \pcode{L}, the \pcode{lockset} of \task{T} is updated to contain lock \pcode{L}. In case of a lock release, \pcode{L} is removed from the lockset.
When \task{T} accesses a shared variable \code{x}, the \pcode{lockset} is copied to the variable metadata. During a race check, \fastpar checks for the intersection of the locksets from the metadata history to infer a data race.
In practice, variables are accessed with the same set of locks, and hence maintaining access metadata for different sets of unique locks is reasonable. Furthermore, metadata update operations depend on the size of the IVC, but we find that the maximum depth of the inheritance tree is small ($\leq 30$ for our benchmarks).

\begin{figure}[t]
    \centering
    \begin{subfigure}{.25\linewidth}
        \centering
        \includegraphics[scale=0.58]{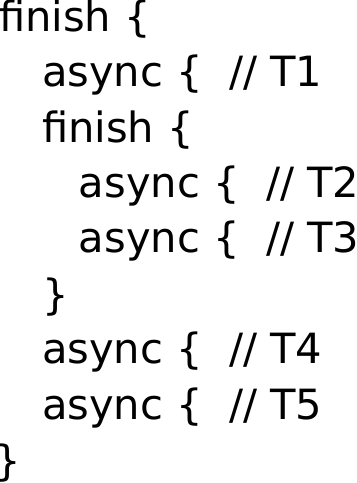}
        \label{fig:example1-code}
    \end{subfigure}%
    \hspace*{1em}
    \begin{subfigure}{.7\linewidth}
        \vspace*{-15pt}
        \centering
        \includegraphics[scale=0.60]{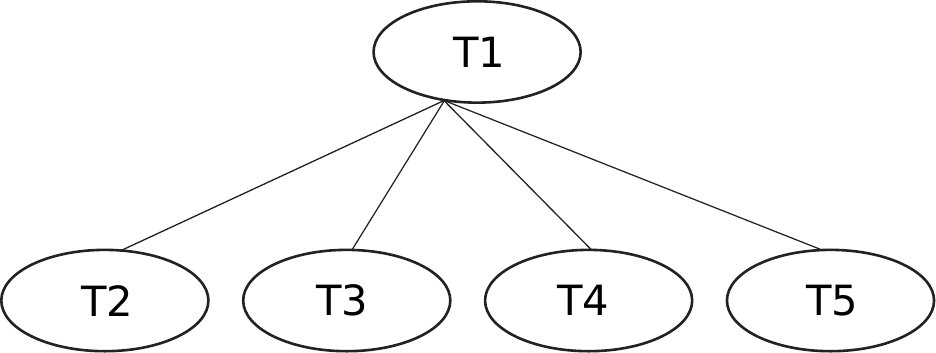}
        \label{fig:example1-ivc}
    \end{subfigure}
    \caption{An \asyncfinish program with five tasks \task{T1}--\task{T5} and the corresponding inheritance tree.}
    \label{fig:tusker-example1}
\end{figure}



\subsection{Characterizing \fastpar}
\label{sec:correctness}

\paragraph*{Race coverage.}

While \FT's race coverage is limited to the observed schedule, \fastpar can detect per-input apparent races like \ptracer. Although both \ft and \fastpar use vector clocks to track happens-before (HB) relations among accesses, \fastpar can detect races in \emph{other} schedules 
due to the following algorithmic differences with \ft.


\begin{itemize}
    \item \ft stores per-thread vector clocks. The races reported by \FT can vary across schedules since many tasks can map to a single thread, and the exact sequence of tasks mapped may vary across schedules. \fastpar stores vector clocks per task, and is not impacted by the mapping of tasks to threads.

    \item \ft tracks the HB relations to establish order among synchronization operations, which is sensitive to the order of lock 
          operations and varies across schedules. \FASTPAR uses the lockset technique to track synchronization operations and stores two reads and two writes per lockset in per-variable metadata. The metadata structure enables \fastpar to detect races irrespective of the order of 
          lock operations, so the number of races reported is the same across schedules.


\end{itemize}

\paragraph*{Correctness.}

The following two lemmas help prove the correctness of \fastpar.
We discuss the proofs in the Appendix.

\begin{lemma}
    \label{lemma:one}
    Storing access entries of only two tasks such that their LCA
    is the highest is sufficient. A future access to the variable, which is racy with any of the existing parallel accesses, must be racy with any one of the two access entries with the highest LCA.
\end{lemma}

\begin{lemma}
    \label{lemma:two}
    The IVC 
    correctly computes the two tasks with the highest LCA among three parallel tasks.
\end{lemma}



\section{Implementation}
\label{sec:impl}


Our implementation extends the
\ptracer artifact.\footnote{\url{https://github.com/rutgers-apl/PTRacer}}
A static compiler pass using LLVM 3.7
instruments
load and store instructions in \CPP programs, and inserts function calls to execute the appropriate dynamic race detection analysis.
The 
implementation uses Intel Threading Building Blocks (TBB)~\cite{tbb-reinders} for
task parallelism.
The public implementation of \ptracer reports wrong race results for the benchmarks \bench{fluidanimate}, \bench{kmeans}, \bench{streamcluster}, and \bench{sort} (see Section~\ref{sec:eval:setup}).
We found
an implementation error was corrupting the DPST built by \ptracer, and there was a race
while updating a global array data structure used for LCA hashing.
After fixing these issues, the race reports were the same across all the tools.
Our modifications have minimal ($\leq$1\%) impact on the performance of \ptracer, and we use our fixed version for the evaluation.
Our prototype implementation of \fastpar extend the same static compiler pass to ensure all the prototypes do the same work. We have also reimplemented the \ft 
algorithm~\cite{fasttrack,fasttrack2}.

\paragraph*{Race detection for fork-join programs}



Utterback \etal propose a parallel and asymptotically optimal algorithm called \WSPORDER for race detection of \forkjoin programs~\cite{cracer-spaa-2016}.
The algorithm uses two order maintenance (OM)
data structures to maintain two total orders 
of all \emph{strands} in the computation. A strand is a sequence of instructions that contain no parallel primitives and executes sequentially.
A strand \pcode{x} logically precedes strand \pcode{y} if and only if \pcode{x} precedes \pcode{y} in both orderings. These orderings are sufficient to determine SP relationships.
The two OM data structures support constant-time operations like \pcode{insert} and \pcode{query}, and \emph{most} concurrent updates do not need synchronization.
However, large parts of the OM data structure are updated during \pcode{relabel} operations and hence require synchronization.
Since \pcode{relabel} operations are serialized,
the algorithm modifies a work-stealing task scheduler to prioritize the
operations. Furthermore, workers blocked on an \pcode{insert} or a \pcode{query} operation help with the \pcode{relabel} instead of being idle.
Note that the WSP-Order algorithm 
does not support lock-based synchronization and requires tight coupling with a work-stealing scheduler for good performance~\cite{cracer-spaa-2016}.





The public implementation of \wsporder is called \cracer.\footnote{\url{https://github.com/wustl-pctg/cracer}}
Since the original implementation uses the Batcher runtime~\cite{batcher-spaa-2014}, we reimplement \cracer with Intel TBB in LLVM
for a fair comparison.
\savespace{Our \cracer implementation instruments the same memory instructions as the other tools, and we try to be as close to the original implementation as possible.}
An important contribution in the original \cracer work is the parallelization of the \pcode{relabel} operations. We have not implemented task scheduler support for parallel \pcode{relabel} operations, \pcode{relabel}s in our implementation are serial.
The total time taken in the serial \pcode{relabel} operations is small in our experiments.
The run time of the benchmarks we report
for \cracer \emph{does not} account for the time taken for the \pcode{relabel} operations, which is a lower bound (\Cref*{sec:eval:perf}). Note that our reimplementation can also differ in performance from the original implementation because of differences in the schedulers.

\smallskip
\noindent We reuse eighty unit test cases 
from the 
\ptracer artifact and designed new test cases to check the correctness of
our \ft, \fastpar, and \cracer implementations. The unit tests include both racy and non-racy
programs, \emph{with and without} locks.
All the prototypes pass the unit tests. Our implementations are publicly available.\footnote{\url{https://github.com/prospar/fastracer-pmam-2022}}

\section{Evaluation\label{sec:eval}}

This section compares \fastpar with the closest prior work, \ft~\cite{fasttrack}, \ptracer~\cite{ptracer-fse-2016}, and \cracer~\cite{cracer-spaa-2016}.


\subsection{Experimental Setup}
\label{sec:eval:setup}

\paragraph*{Benchmarks}

We reuse twelve TBB-based applications used by \ptracer for our evaluation. These include four
applications, \blackscholes, \fluidanimate, \streamcluster, and \swaptions, from the PARSEC benchmark suite~\cite{parsec-pact-2008},
five geometry and graphics applications, \convexhull, \delrefine, \deltriang, \nn, and \raycast, from the PBBS benchmark suite~\cite{pbbs-benchmarks}, and three applications, \karatsuba, \kmeans, and \sort, from the Structured Parallel Programming book~\cite{spp-book}.
We left out the PARSEC application, \bodytrack, because of a compilation error,
and ignore the \CRACER benchmarks because they use Cilk-5~\cite{cracer-spaa-2016}.

The benchmarks follow \pcode{spawn-sync} semantics where a child task joins with its immediate parent (\pcode{async-finish} semantics are more general) and do not use locks, so we were able to run \cracer successfully for all the benchmarks.





\paragraph*{Evaluation platform}







The experiments execute on an Intel Xeon Gold 5218 system with one 16-core processor with hyperthreading disabled, 128~GB DDR4 primary memory, running Ubuntu Linux 20.04.3 LTS with kernel version 5.11.0.

\begin{figure*}
    \centering
    \includegraphics[width=\textwidth]{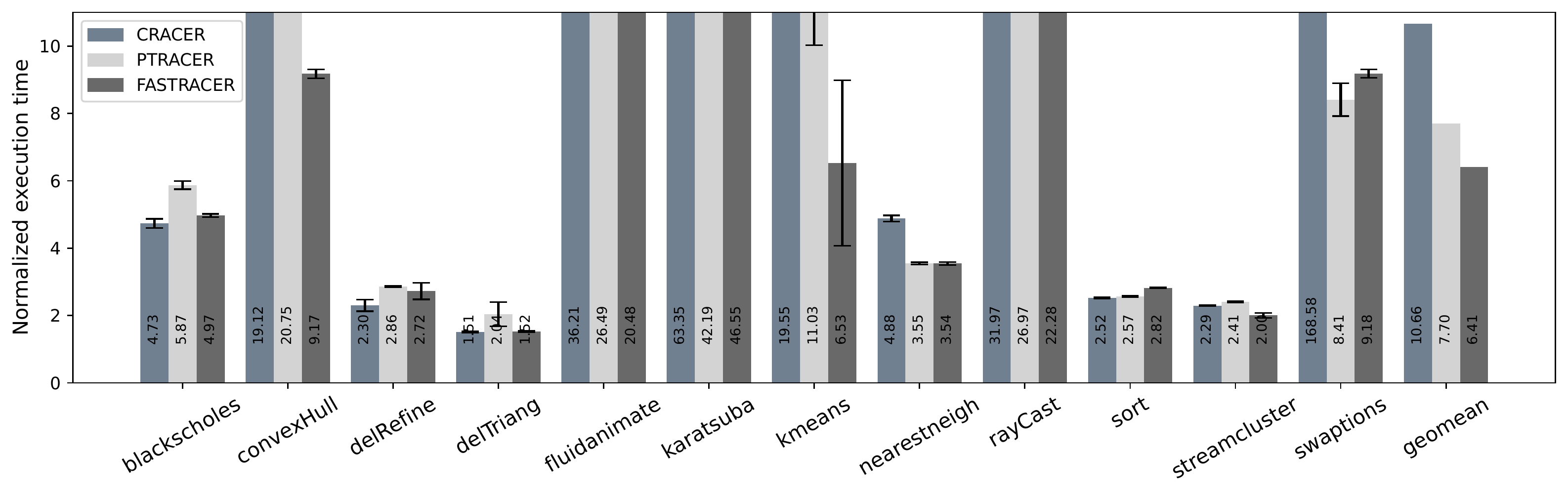}
    \caption{Performance comparison of the different techniques normalized to the unmodified execution time of the benchmarks.}
    \label{fig:aravalli-norm-bench-time}
\end{figure*}

\subsection{Performance Results\label{sec:eval:perf}}

Figure~\ref{fig:aravalli-norm-bench-time} reports the performance of \cracer, \ptracer, and \fastpar for all the benchmarks (arranged alphabetically). Every bar averages ten trials and is normalized to the baseline, which runs the unmodified benchmarks without instrumentation. Smaller bars mean better run times.
By default, the TBB scheduler creates $c$ threads and multiplexes the tasks to the $c$ threads, where $c$ is the number of cores in the system.
The results in Figure~\ref{fig:aravalli-norm-bench-time} are with 16 threads.

Our reimplementation of \cracer incurs an overhead of \cracernotoolperfoverheadaravalli over the unmodified applications.
The average performance slowdown incurred by \ptracer is \ptracernotoolperfoverheadaravalli, while it is 
\fastparnotoolperfoverheadaravalli for \fastpar.
\FastPar shows substantial improvement for multiple benchmarks like \convexhull, \deltriang, \fluidanimate, \kmeans, and \raycast.
The better performance of \fastpar is due to improved locality from vector clocks compared to cache-unfriendly tree traversals in \cracer and \ptracer.
On average, \fastpar outperforms both \cracer and \ptracer by \fastparcracerspeeduparavalli and \fastparptracerspeeduparavalli, respectively.



We cannot directly compare \cracer results to the original work~\cite{cracer-spaa-2016} because of differences in the implementation, benchmarks, and evaluation infrastructure. However, we reemphasize that our \cracer results do not include the cost of \pcode{relabel} operations.
The original \PTRACER work reports an average slowdown of 6.7X, compared to \ptracernotoolperfoverheadaravalli in our experiments.
In addition to the differences in the evaluation platforms and bug fixes, different execution times of the benchmarks in the absence of perfect scaling will lead to different relative overheads.
We have verified that our modifications have minimal ($\leq$1\% difference) impact on the performance of \ptracer.


\FT successfully executed with four benchmarks that required less memory, \blackscholes, \fluidanimate, \karatsuba, and \sort, and got killed on the other benchmarks.
The overhead of \FT for the four benchmarks is \ftnotoolperfoverheadaravalli, and \fastpar outperforms \ft by \ftfastparspeeduparavalli.
Given that \fastpar builds on \ft,
this result shows that the
task and variable vector clock
optimizations proposed in \fastpar are effective in reducing both run time and space overheads.

\begin{small}
    \begin{table}
        \centering
        \begin{tabular}{lrrrrr}
                                  & \textbf{UM} & \textbf{FT} & \textbf{CR} & \textbf{PT} & \textbf{FR} \\
            \toprule
            \bench{blackscholes}  & 0.62        & 9.64        & 10.38       & 54.78       & 48.03       \\
            \bench{fluidanimate}  & 0.51        & 69.93       & 0.58        & 9.17        & 2.04        \\
            \bench{streamcluster} & 0.84        & -           & 1.17        & 22.76       & 6.91        \\
            \bench{swaptions}     & 0.19        & -           & 0.98        & 17.34       & 3.87        \\
            \midrule
            \bench{convexHull}    & 1.69        & -           & 2.97        & 19.94       & 12.71       \\
            \bench{delRefine}     & 3.71        & -           & 6.25        & 125.37      & 125.39      \\
            \bench{delTriang}     & 3.63        & -           & 8.19        & 113.69      & 106.77      \\
            \bench{nearestneigh}  & 2.19        & -           & 5.60        & 106.86      & 125.26      \\
            \bench{rayCast}       & 1.28        & -           & 4.08        & 18.37       & 14.35       \\
            \midrule
            \bench{karatsuba}     & 0.02        & 0.86        & 0.18        & 9.34        & 2.81        \\
            \bench{kmeans}        & 0.11        & -           & 2.52        & 15.82       & 5.99        \\
            \bench{sort}          & 0.15        & 18.96       & 0.20        & 9.87        & 2.40        \\
        \end{tabular}
        \caption{Comparison of the memory overhead (in GBs).%
        }
        \label{tab:memory}
    \end{table}
\end{small}

\paragraph*{Memory overhead}

Table~\ref{tab:memory} compares the peak memory requirement of each benchmark with the four techniques, as reported by the Massif tool in Valgrind~\cite{valgrind}.
Column 2, \emph{UM}, shows the memory requirement of the unmodified
application, while \emph{FT}, \emph{CR}, \emph{PT}, and \emph{FR} stand for \ft, \cracer, \ptracer, and \fastpar,
respectively. The results show that using IVCs for encoding task inheritance in \fastpar compared to
the fine-grained DPST structure in \ptracer provides significant memory savings in maintaining
per-task metadata, especially for \fluidanimate, \karatsuba, \kmeans, \sort, and \swaptions. 
\CRACER usually requires the least memory since it maintains order maintenance structures, but has
high performance overhead compared to \fastpar.

\begin{figure*}[ht]
    \centering
    \begin{subfigure}{.24\textwidth}
        \centering
        \includegraphics[width=1\linewidth]{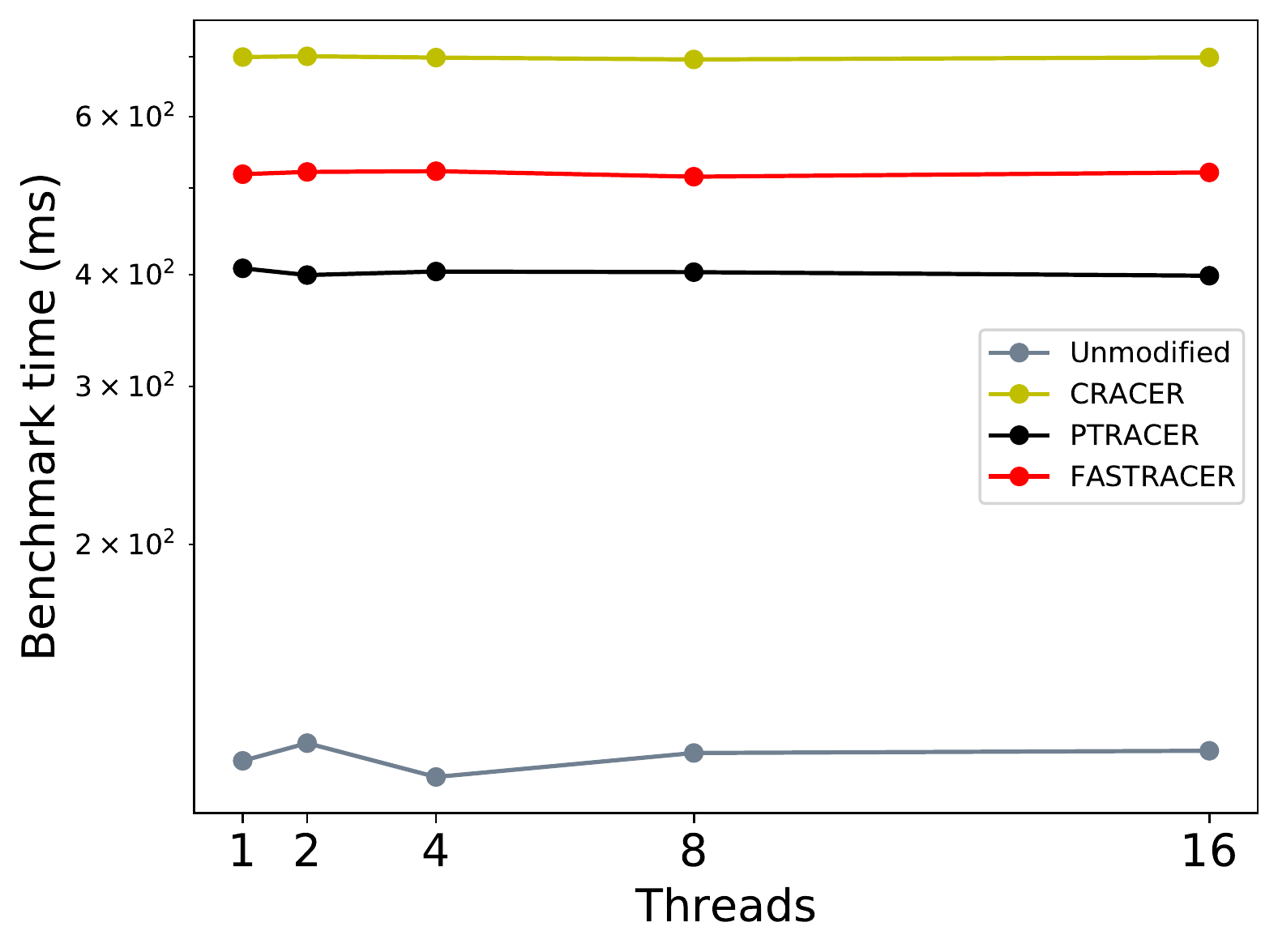}
        \caption{\blackscholes}
        \label{fig:blackscholes}
    \end{subfigure}%
    \begin{subfigure}{.24\textwidth}
        \centering
        \includegraphics[width=1\linewidth]{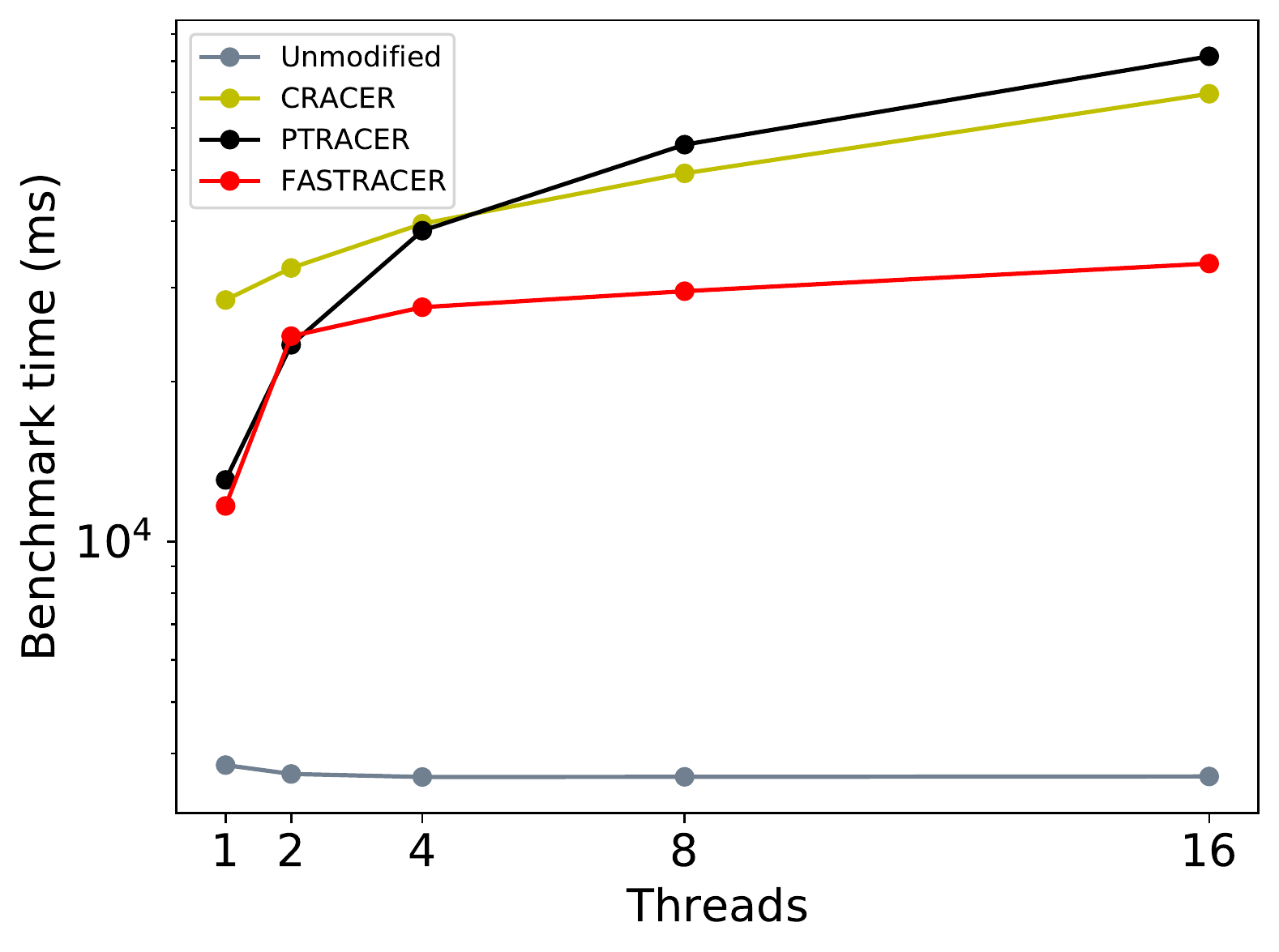}
        \caption{\convexhull}
        \label{fig:convexhull}
    \end{subfigure}
    \begin{subfigure}{.24\textwidth}
        \centering
        \includegraphics[width=1\linewidth]{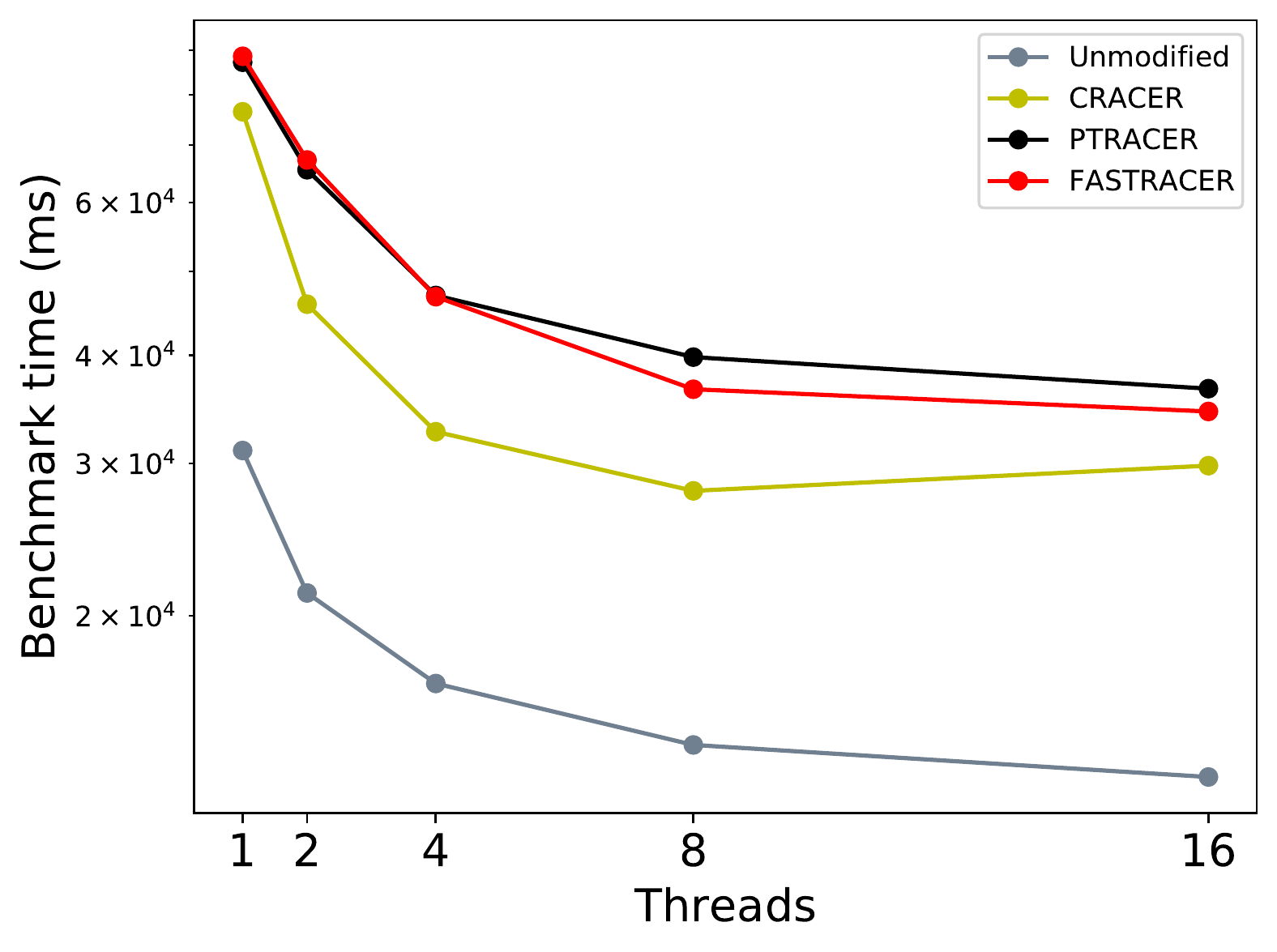}
        \caption{\delrefine}
        \label{fig:delrefine}
    \end{subfigure}
    \begin{subfigure}{.24\textwidth}
        \centering
        \includegraphics[width=1\linewidth]{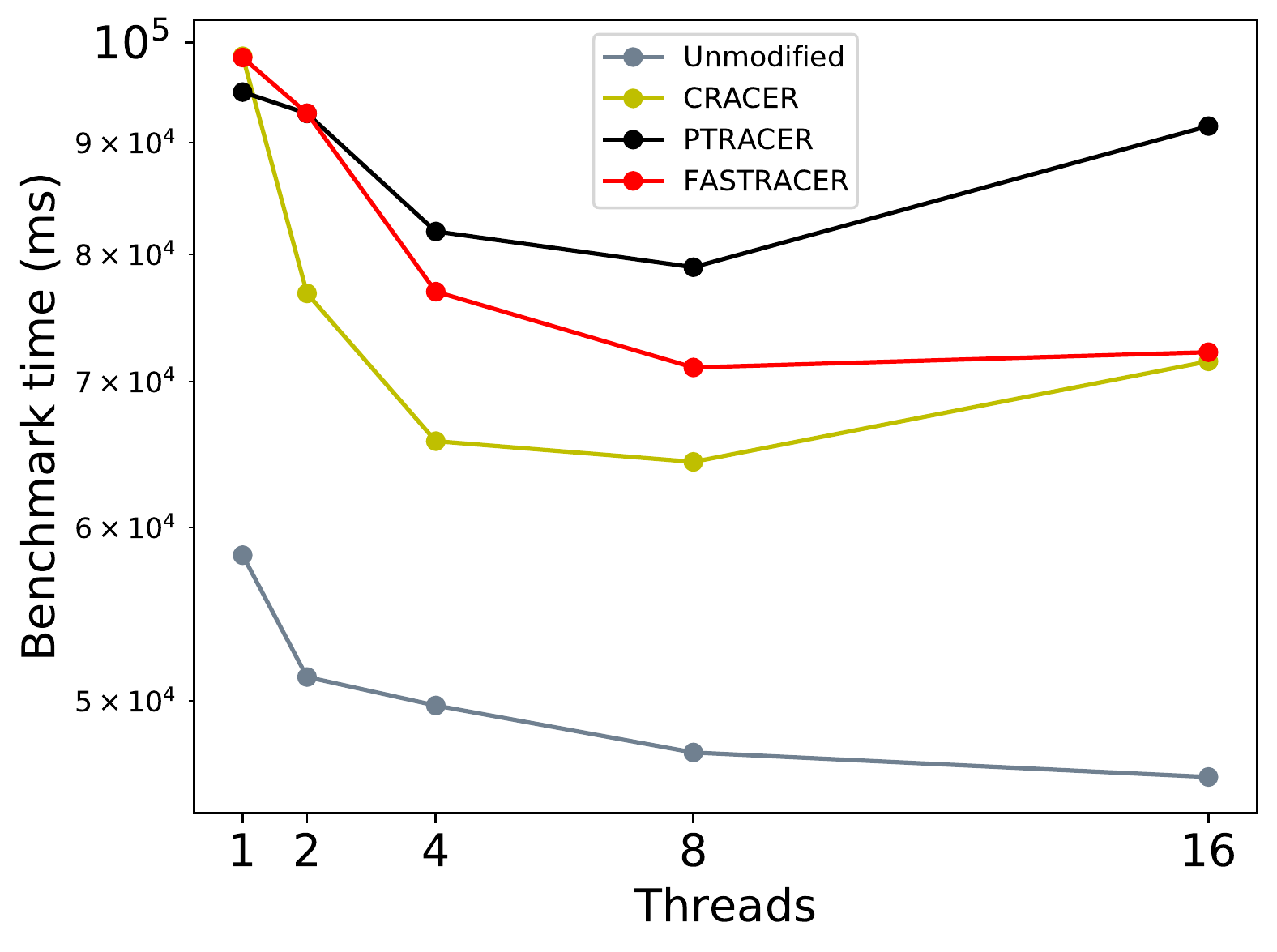}
        \caption{\deltriang}
        \label{fig:deltriang}
    \end{subfigure}\\
    \begin{subfigure}{.24\textwidth}
        \centering
        \includegraphics[width=1\linewidth]{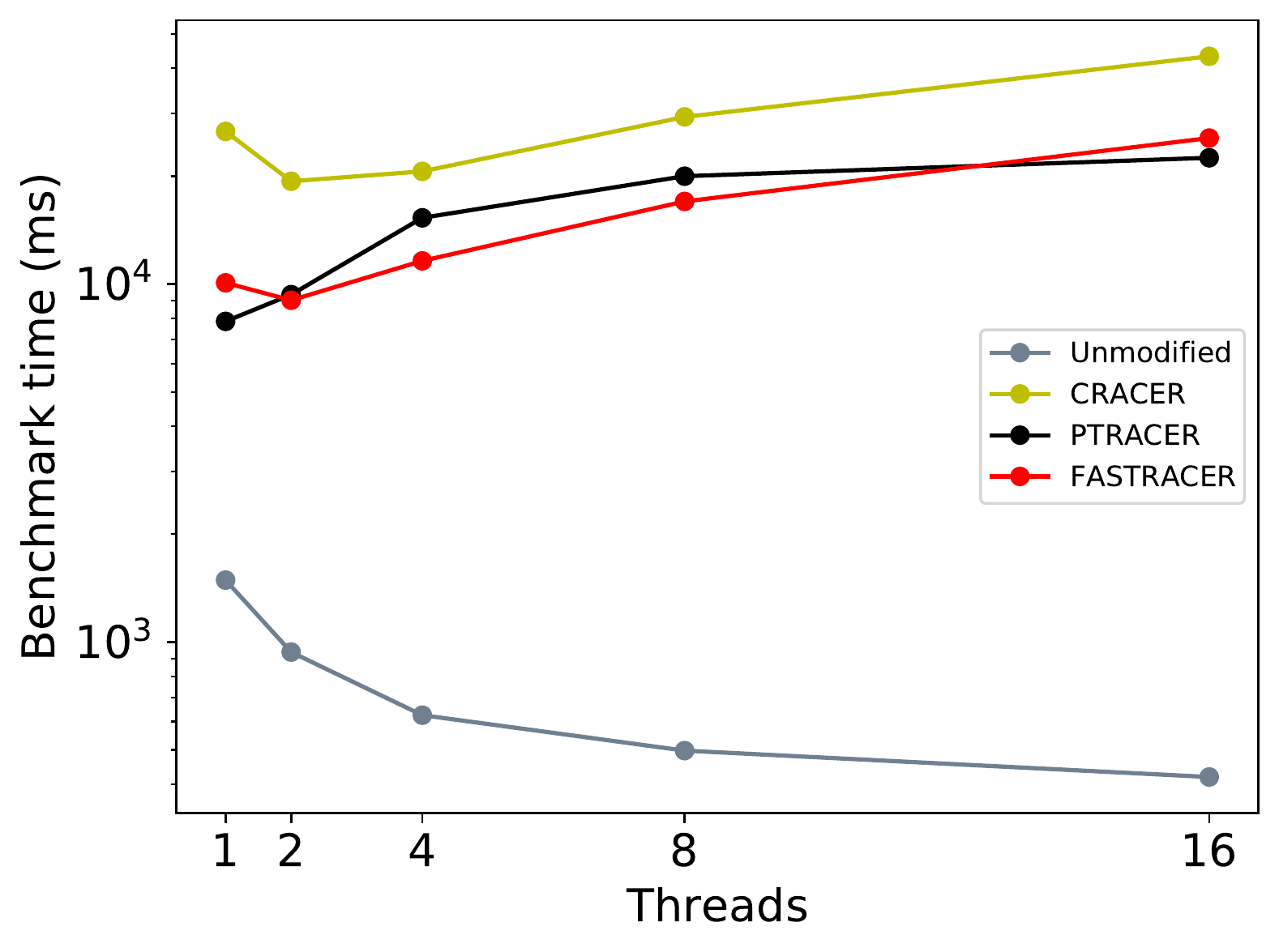}
        \caption{\fluidanimate}
        \label{fig:fluidanimate}
    \end{subfigure}
    \begin{subfigure}{.24\textwidth}
        \centering
        \includegraphics[width=1\linewidth]{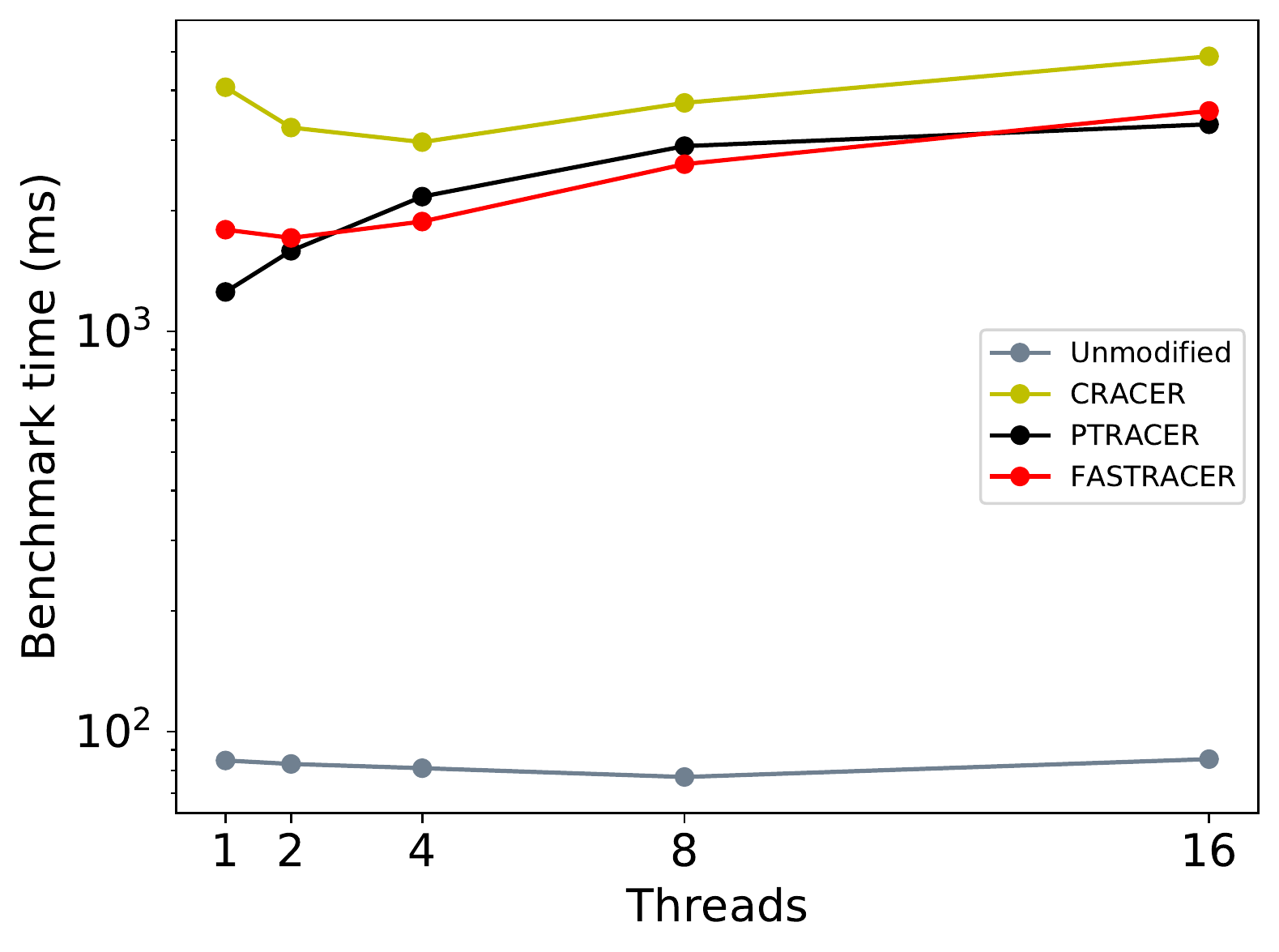}
        \caption{\karatsuba}
        \label{fig:karatsuba}
    \end{subfigure}
    \begin{subfigure}{.24\textwidth}
        \centering
        \includegraphics[width=1\linewidth]{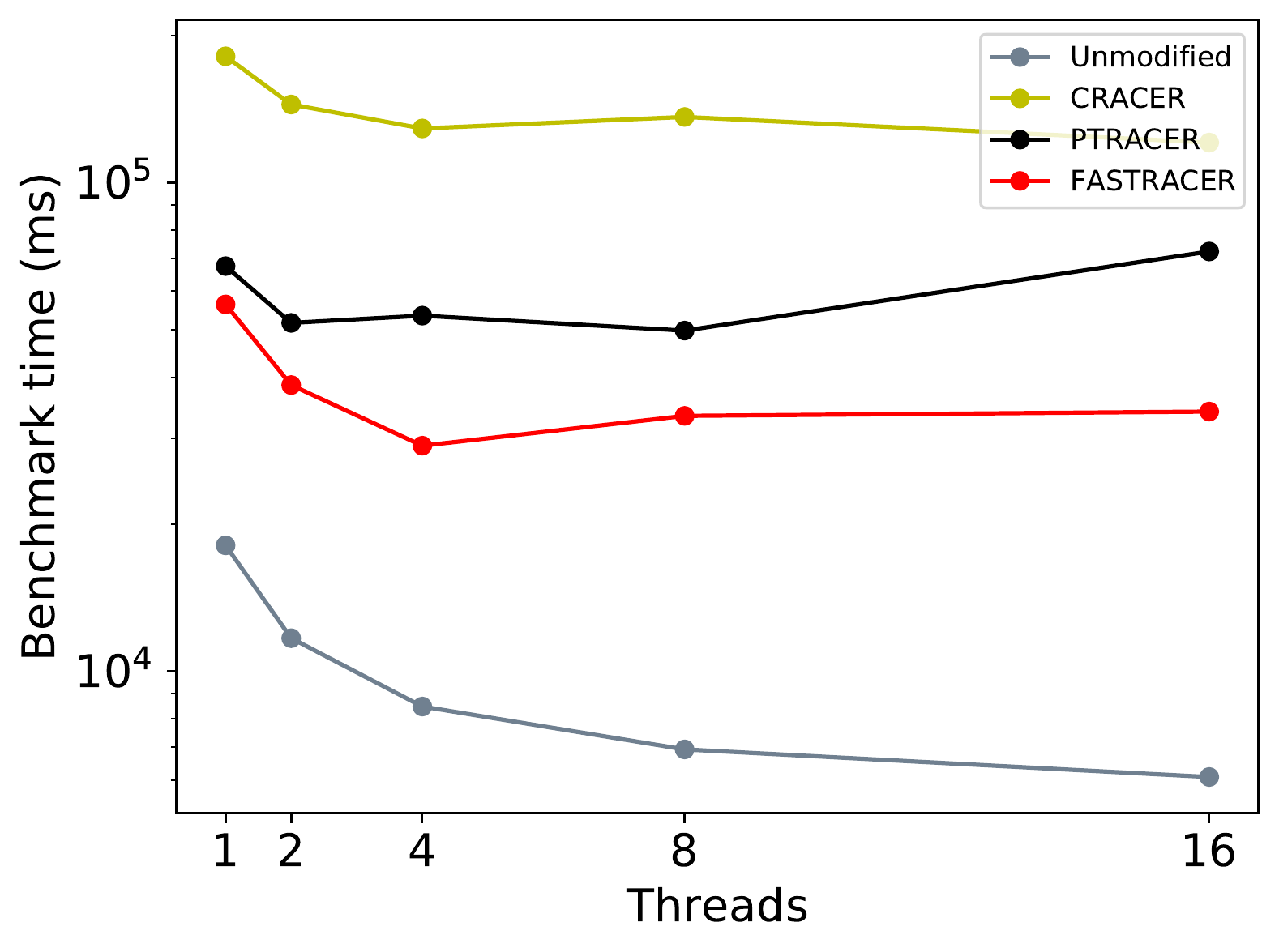}
        \caption{\kmeans}
        \label{fig:kmeans}
    \end{subfigure}
    \begin{subfigure}{.24\textwidth}
        \centering
        \includegraphics[width=1\linewidth]{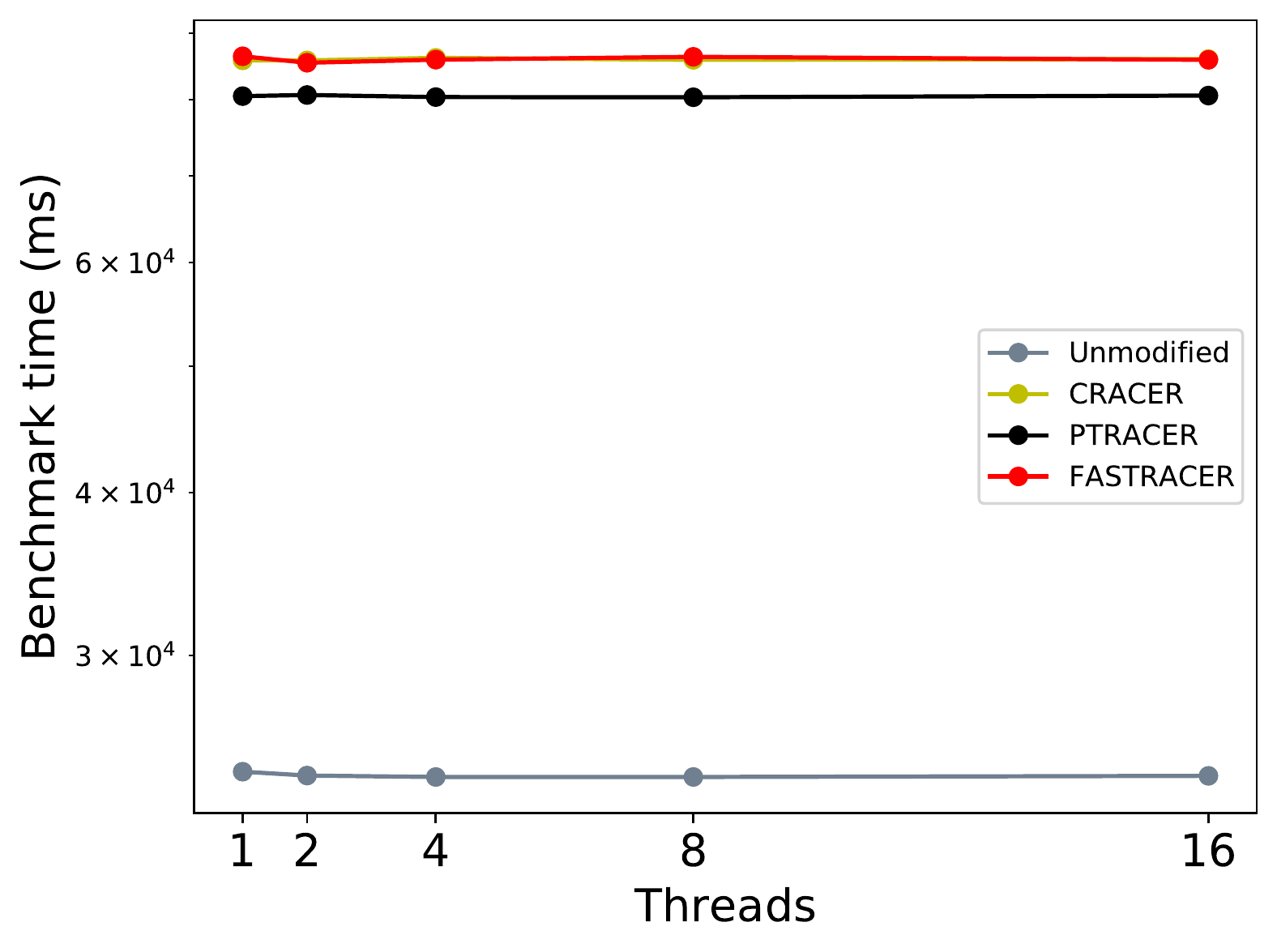}
        \caption{\nn}
        \label{fig:nn}
    \end{subfigure}\\
    \begin{subfigure}{.24\textwidth}
        \centering
        \includegraphics[width=1\linewidth]{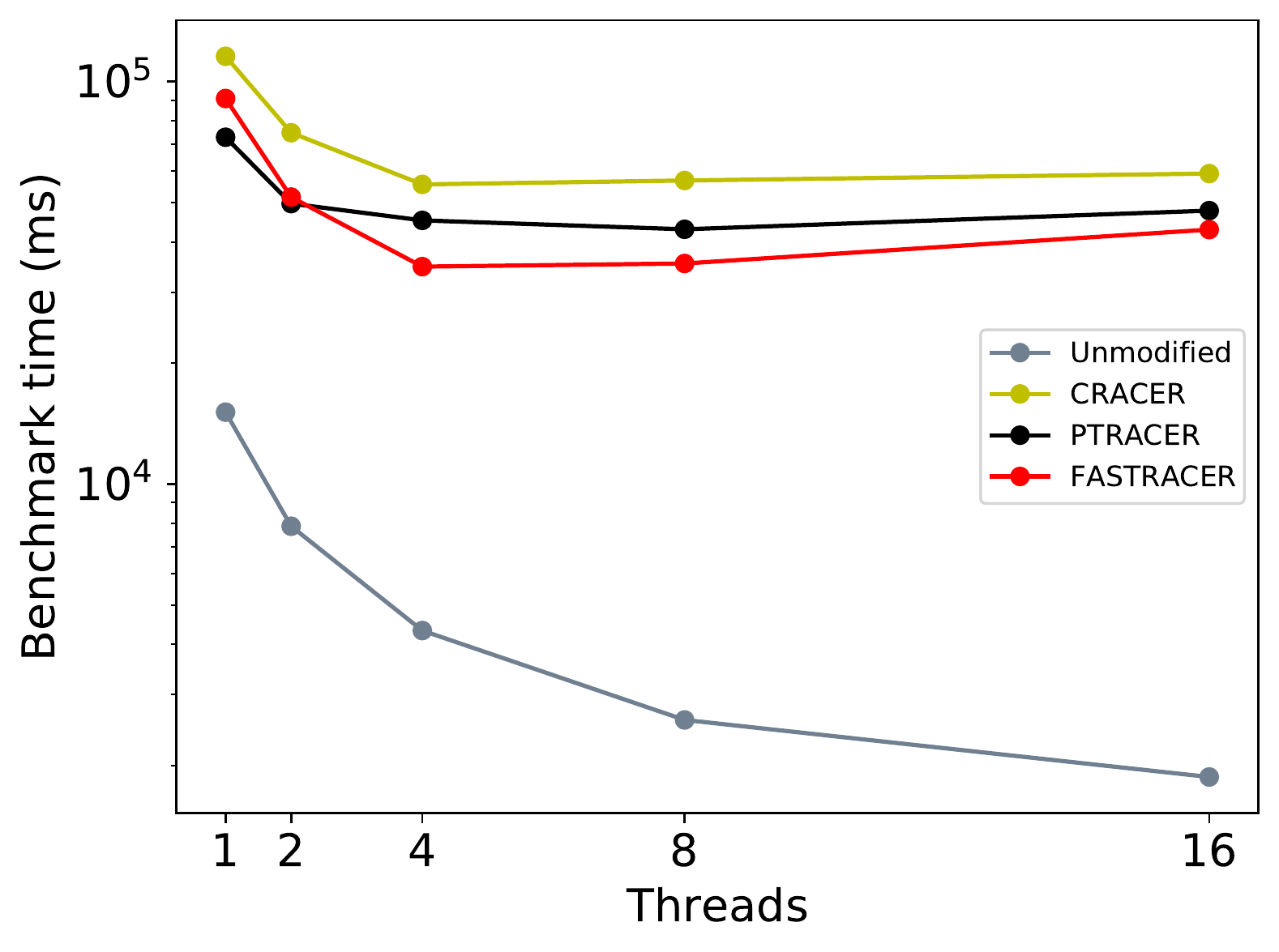}
        \caption{\raycast}
        \label{fig:raycast}
    \end{subfigure}
    \begin{subfigure}{.245\textwidth}
        \centering
        \includegraphics[width=1\linewidth]{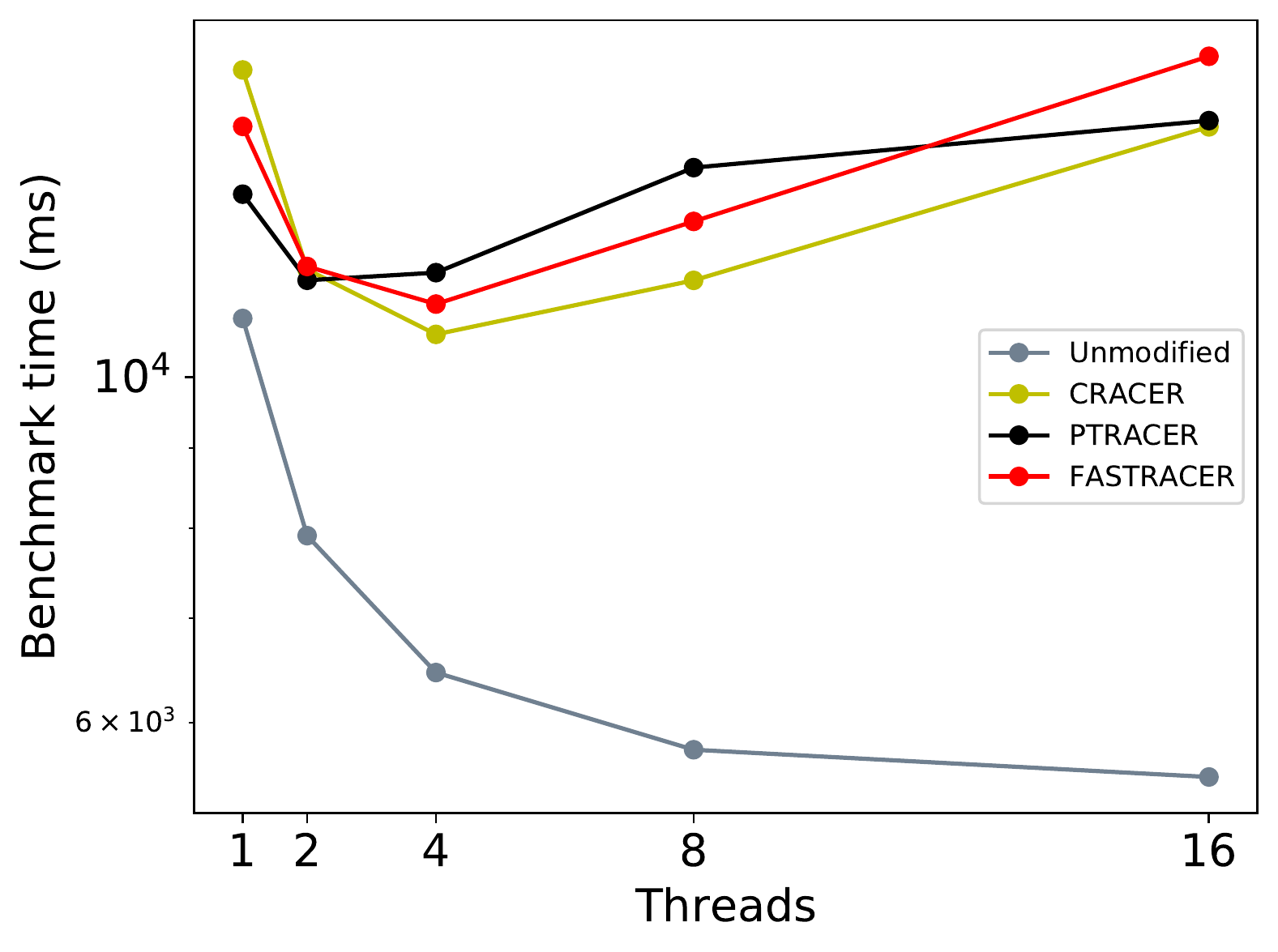}
        \caption{\sort}
        \label{fig:sort}
    \end{subfigure}
    \begin{subfigure}{.245\textwidth}
        \centering
        \includegraphics[width=1\linewidth]{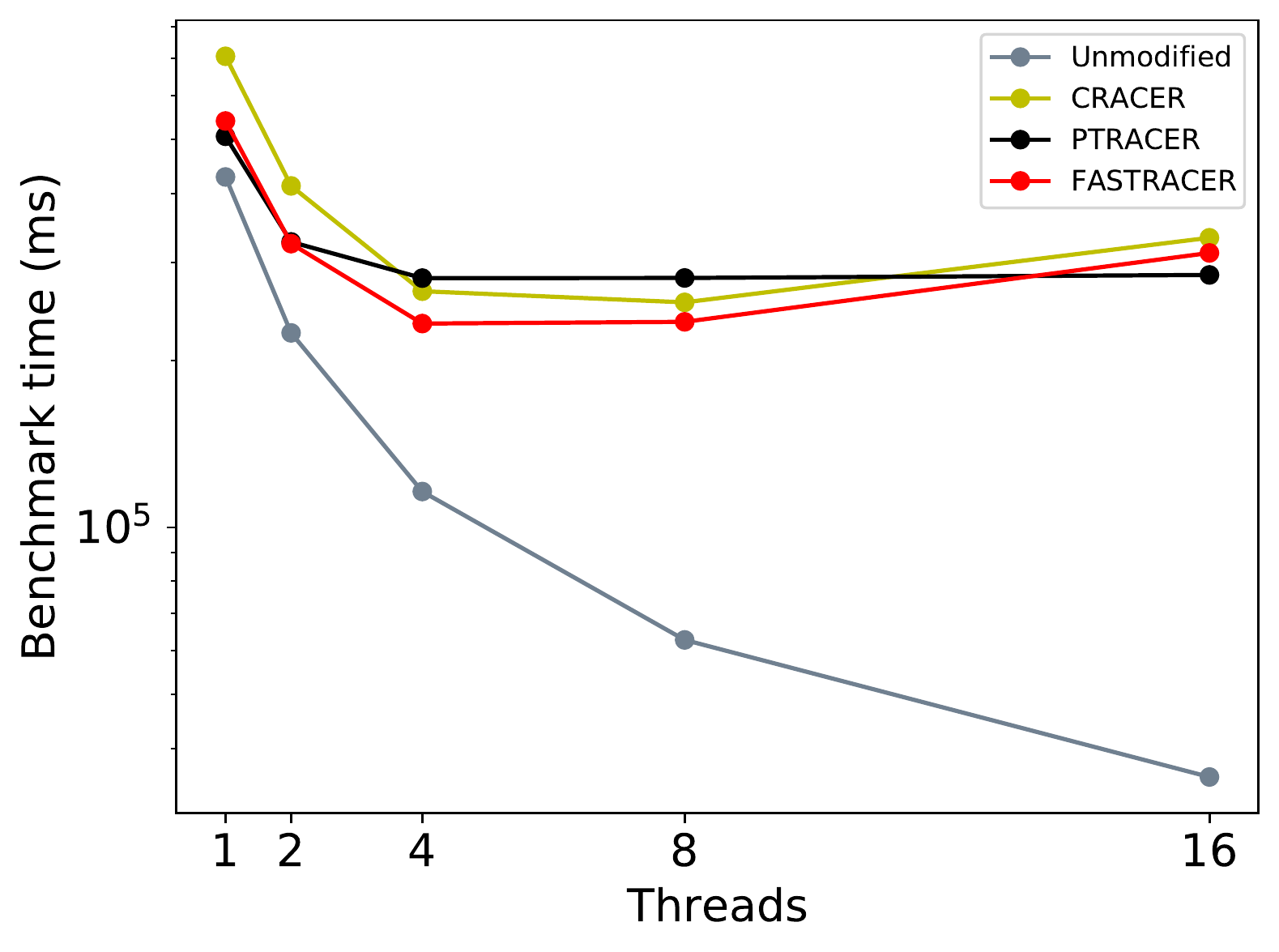}
        \caption{\streamcluster}
        \label{fig:streamcluster}
    \end{subfigure}
    \begin{subfigure}{.245\textwidth}
        \centering
        \includegraphics[width=1\linewidth]{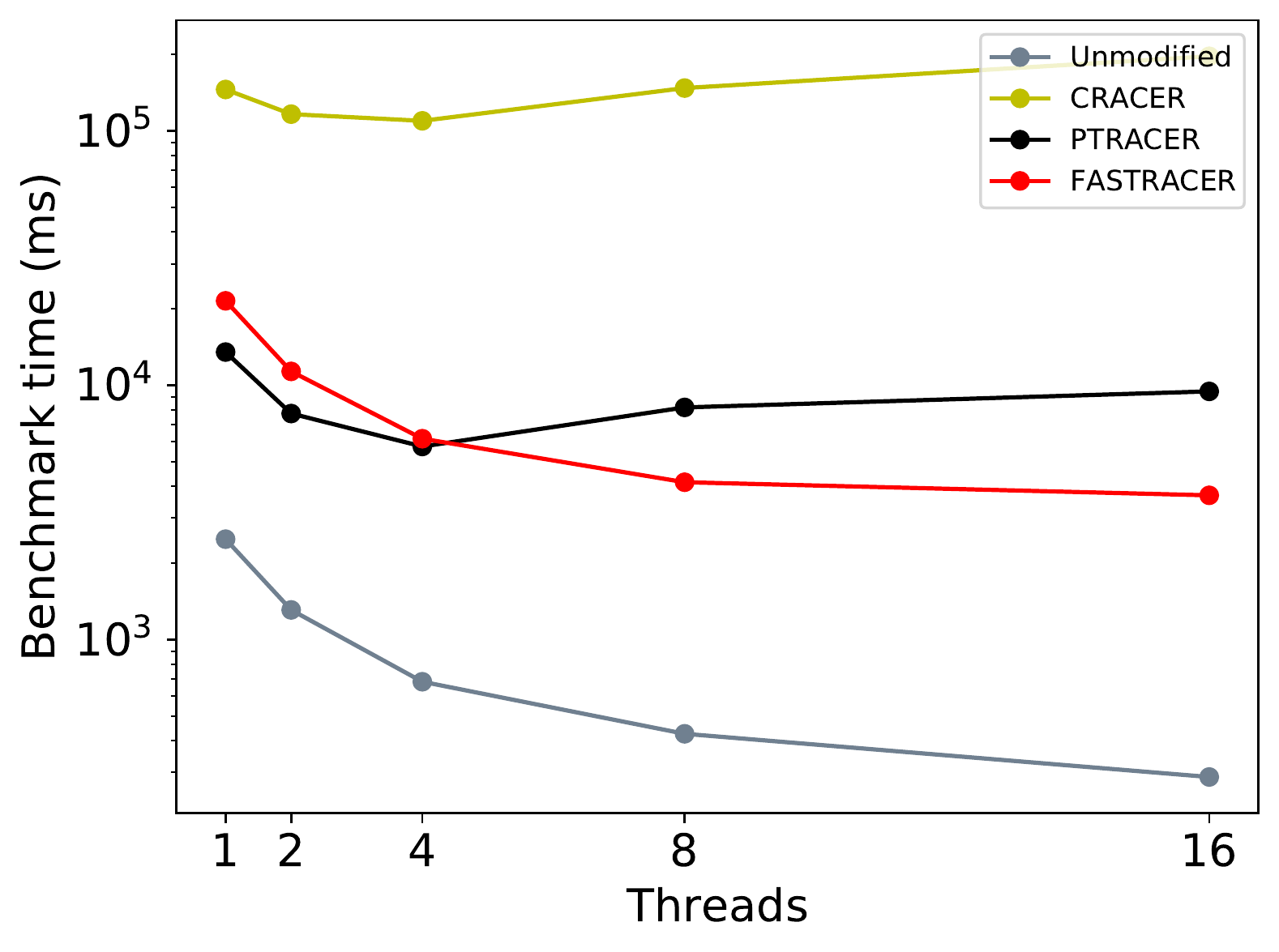}
        \caption{\swaptions}
        \label{fig:swaptions}
    \end{subfigure}
    \caption{Scalability results of the benchmarks on the Intel Gold platform described in \Cref{sec:eval:setup}.}
    \label{fig:scalability}
\end{figure*}

\paragraph*{Scalability}

Figure~\ref{fig:scalability} shows the scalability plots of the benchmarks as we vary the number of threads (in powers of two) used by the TBB scheduler.
Other experimental setup are the same as in Figure~\ref{fig:aravalli-norm-bench-time}.
The \emph{Unmodified} configuration in the figure
shows that most applications scale well, excepting
\blackscholes, \convexhull, and \karatsuba.
\FastPar scales better than \ptracer for \deltriang, \kmeans, and \swaptions, for the range of thread counts we have evaluated.
The scaling behavior of \cracer, \PTRACER, and \fastpar are very similar for the remaining benchmarks.



\paragraph*{Platform sensitivity}


We evaluate the sensitivity of the optimizations by re-running experiments
on an Intel Xeon Silver 4114 system with two ten core processors
(twenty 2.0~GHz cores in total)
with hyperthreading turned off, 128~GB primary memory, running Ubuntu Linux 18.04.6
LTS with kernel version 4.15.0.
The rest of the methodology is the same as Figure~\ref{fig:aravalli-norm-bench-time}.
\CRACER, \ptracer, and \fastpar incur overheads of \cracernotoolperfoverheadhimalaya,
\ptracernotoolperfoverheadhimalaya, and \fastparnotoolperfoverheadhimalaya, respectively.
Although, the overheads on the multi-socket Xeon Silver platform are slightly greater compared to the Xeon Gold platform, 
the general trend is the same.
\fastpar outperforms both \cracer and \ptracer by
\fastparcracerspeeduphimalaya and \fastparptracerspeeduphimalaya, respectively.
The study confirms the effectiveness of the proposed optimizations \emph{across} different platforms.

\medskip
\noindent \FastPar outperforms prior work on the \emph{same} set of benchmarks as used in \ptracer. Given the generic nature of the optimizations, we expect a similar qualitative trend for benchmarks where threads join with arbitrary ancestors (\eg, \asyncfinish semantics). More importantly, \fastpar disproves the assumption made in \emph{all} prior work that vector-clock-based analysis is not suited for task-based programs.

\subsection{Data Races and Run-time Statistics}
\label{sec:eval:stats}

Table~\ref{tab:run-time-stats} summarizes the run-time statistics. The data is the average from 10
trials with a statistic-collecting configuration. Columns 2--4 in the table show the average number
of tasks spawned by the benchmarks and the number of read and write accesses instrumented
by the static compiler pass in LLVM. Columns 5--7 show the number of data races reported by the
different tools. To stress-test the correctness of our implementations, we introduced data races in
a few benchmarks that already did not have known races. The suffix ``-r'' denotes benchmarks that
have been modified to introduce races for evaluation. All the tools are expected to report the same
number of races for a given application with a fixed input. \CRACER (CR), \PTRACER (PT), and
\fastpar (FR) report the same violations for all the benchmarks. Please note that the
evaluation includes our fixes to \ptracer (\Cref*{sec:impl}).

\begin{small}
    \begin{table}
        \setlength{\tabcolsep}{3pt}
        \centering
        \begin{tabular}{lrrrrrrr}
                            & {\textbf{\# Tasks}} & \multicolumn{2}{c}{\textbf{ACC} ($\times 10^{6}$)} & \multicolumn{3}{c}{\textbf{Data Races}}                                           \\
            \cline{3-4}
            \cline{5-7}
                            & ($\times 10^{3}$)   & \textbf{RDs}                                       & \textbf{WRs}                            & \textbf{CR} & \textbf{PT} & \textbf{FR} \\
            \toprule

            \blackscholes   & 0.20                & 90                                                 & 50                                      & 21          & 21          & 21          \\


            \fluidanimater  & 1.60                & 26                                                 & 0.7                                     & 40          & 40          & 40          \\

            \streamclusterr & 180                 & 363                                                & 13                                      & 80          & 80          & 80          \\

            \swaptions      & 960                 & 77                                                 & 77                                      & 0           & 0           & 0           \\

            \midrule

            \convexhull     & 8.50                & 30                                                 & 0                                       & 0           & 0           & 0           \\

            \delrefine      & 1000                & 15                                                 & 0                                       & 0           & 0           & 0           \\

            \deltriang      & 790                 & 30                                                 & 20                                      & 0           & 0           & 0           \\

            \nn             & 2800                & 51                                                 & 8                                       & 0           & 0           & 0           \\

            \raycast        & 1900                & 160                                                & 0                                       & 0           & 0           & 0           \\

            \midrule

            \karatsuba      & 1.98                & 3.4                                                & 0.8                                     & 0           & 0           & 0           \\

            \kmeans-r       & 35                  & 570                                                & 10                                      & 75          & 75          & 75          \\

            \sort           & 0.70                & 11                                                 & 0.06                                    & 1024        & 1024        & 1024        \\
        \end{tabular}
        \caption{Run-time statistics across different benchmarks.}
        \label{tab:run-time-stats}
    \end{table}
\end{small}


\section{Related Work}
\label{sec:relatedwork}

In the following, we discuss other related work that has not already been discussed.

\paragraph*{Race detection for task-based programs}

Mellor-Crummey 
exploited the structural property of fork-join programs to show that tracking two readers and a single writer per memory location is sufficient
for sound data race detection~\cite{mellor-crumney-sc-1991}. Since then, there has been much work to design 
race detection algorithms to utilize the serial-parallel (SP) structure of programs with constant space overhead for metadata~\cite{spbags,cilk-spaa-98}.
ESP-bags is an extension to SP-bags that supports the \pcode{finish} construct in \asyncfinish programs~\cite{esp-bags-2010}. However, these approaches constrain the program to execute serially in depth-first order, which does not scale. \textsc{Tardis} does not keep track of the SP 
relationships among program strands~\cite{tardis-pldi-2014}. Instead, \textsc{Tardis} maintains log-based access sets and lazily detects races by checking for overlapping intersections of access sets of logically parallel sub-computations at join points.


\paragraph*{Race detection for multithreaded programs}


Static analysis can potentially detect all \emph{feasible} data races across all
possible executions (\ie, no false negatives), but usually do not scale to large programs and
suffer from false positives, which developers loathe
~\cite{naik-static-racedet-2007,naik-static-racedet-2006,racerd}.
Dynamic analyses have the potential to be sound and
precise for the \emph{observed} executions.
Many dynamic data race detection analyses track the happens-before relation to infer data races~\cite{pacer-pldi-2010,fib-oopsla-2017,google-tsan-v1,google-tsan-v2,multirace}.
\emph{Lockset} analysis reports data races when a locking discipline is violated, but can report
false races~\cite{eraser,object-racedet-2001,choi-racedet-2002,ocallahan-hybrid-racedet-2003}. Hybrid techniques integrate both HB and lockset analysis~\cite{ocallahan-hybrid-racedet-2003,racetrack}, but continue to suffer from the
disadvantages of both techniques.
Other techniques sacrifice soundness for performance by sampling memory
accesses~\cite{datacollider,racechaser-caper-cc-2017,pacer-pldi-2010,literace} or require hardware
support to speed up the race detection analysis~\cite{radish,lard,parsnip}.
Prior analyses have also explored detecting a subset of data races that arise due to conflicts among overlapping regions~\cite{conflict-exceptions,valor-oopsla-2015,ifrit}.
Data race detection analyses have also been proposed for parallel programming models such as OpenMP~\cite{archer,romp-sc-2018}
and GPUs~\cite{barracuda-pldi-2017,scord-isca-2020,iguard-sosp-21}.

\paragraph*{Improve race detection coverage}

Recent work has explored exposing as many true races as possible by observing one or a few dynamic
executions. The idea is to
explore different valid interleavings to maximize detecting data races that can occur across
schedules. Many techniques perturb the execution in an attempt to break spurious HB
relations~\cite{racageddon,rvpredict-pldi-2014,racefuzzer}. \emph{Predictive}
techniques aim to detect data races that can occur in other 
correct
reorderings of memory accesses by observing one dynamic execution.
Such techniques use partial order relations that are weaker than HB to allow reconstructing other valid memory
reorderings~\cite{cp-popl-2012,vindicator-pldi-2018,wcp-pldi-2017,smarttrack}.





\section{Conclusion\label{sec:conclusion}}





Whereas prior work has overlooked the possibility of using vector clocks for task-based programs~\cite{cracer-spaa-2016,ptracer-fse-2016,spd3-pldi-2012}, \fastpar shows how to achieve competitive performance
with a curated set of optimizations.
\FastPar introduces novel optimizations that reduce the metadata redundancy and alleviate the performance bottlenecks with task vector clocks.
\FASTPAR also exploits the structured execution of \asyncfinish programs to limit the per-variable metadata overhead by using a coarse-but-efficient encoding of task inheritance relationships.
\FASTPAR shows substantial performance improvements over \ft,
and outperforms state-of-art approaches \cracer and \ptracer.
\FastPar's proposed vector clock optimizations allow for efficient dynamic data race detection for task-based \asyncfinish programs.


\begin{acks}
    We thank Michael D. Bond for feedback on the text,
    the \ptracer authors for help with their implementation,
    and Vivek Kumar for discussions and advice.

    This material is based upon work supported by IITK Initiation Grant and \grantsponsor{GS100000001}{Science and Engineering Research Board}{} under Grant \grantnum{GS100000001}{SRG/2019/000384}.
\end{acks}

\iftoggle{ieeeFormat}{
    \section*{Acknowledgment}
    We thank Adarsh Yoga, Santosh Nagarakatte, and Aarti Gupta for their help with the \ptracer implementation.
    Thank Mike Bond and Vivek Kumar.
}{}


\iftoggle{acmFormat}{
    \bibliographystyle{ACM-Reference-Format}
}{}
\iftoggle{ieeeFormat}{
    \balance
    \bibliographystyle{IEEEtran}
}{}
\iftoggle{lncsFormat}{
    \bibliographystyle{./europar22/splncs04}
}{}
\bibliography{./references/venue-abbrv,./references/references}

\iftrue
    \iftoggle{acmFormat}{
        \appendix
        \section{Appendix}

\subsection{Correctness of \fastpar}


\begin{lemma}
    Storing access entries of only two tasks such that their LCA using the inheritance tree is the highest is sufficient. That is, a future access to the variable, which is racy with any of the existing parallel accesses, must be racy with any one of the two access entries with the highest LCA.
\end{lemma}


\begin{proof}
    Let tasks \task{T1}, \task{T2}, and \task{T3} be three tasks that access shared variable \code{var} in parallel with the same lockset. Let the two tasks with the highest LCA in the inheritance tree be \task{T1} and \task{T2}, their LCA be task \task{P}, and \task{T3} be any task other than \task{T1} and \task{T2}.
    All the parallel tasks, including \task{T1}, \task{T2}, and \task{T3} must be immediate or recursive children of \task{P}, \ie, \task{P} is the common ancestor of all \task{T1},\task{T2} and \task{T3}. An immediate child means that the task has been spawned by \task{P} and a recursive child means that the task has been spawned by some other task which itself is an immediate or recursive child of \task{P}.

    Let us assume that in the future, another task \task{T} accesses the shared variable \code{var} and the access is racy with task \task{T3}, but not with \task{T1} and \task{T2}. This implies the access to \code{var} by \task{T} \emph{happens after} both \task{T1} and \task{T2}. Since \task{T1}, \task{T2}, and \task{T3} have the same lockset, and \task{T} is racy with \task{T3}, so the intersection of their locksets must be empty. This implies that we have two non-racy accesses from tasks \task{T1} and \task{T} to \code{var} even when the accesses are not protected by the same lock. This can only happen if \task{T1} and \task{T} are ordered by task management constructs. That is, one of the following must hold:
    (i) \task{T1} accesses \code{var}, and then spawns \task{T}, either immediate or recursive, or (ii) \task{T1} is present in \code{joined} of \task{T}.
    By a similar argument, one can draw a similar conclusion for \task{T2}, \ie, one of the following must hold: (i) \task{T2} accesses \code{var} and then spawns \task{T}, either immediate or recursive, or (ii) \task{T2} is present in \code{joined} of \task{T}.

    Upon careful observation, one can see that constraints for \task{T1} and \task{T2} can hold simultaneously if and only if both \task{T1} and \task{T2} are present in \code{joined} of \task{T}. Other possibilities lead to
    contradictions. This can only happen if some common ancestor \task{P}$'$ of tasks \task{T1} and \task{T2} calls join, in which all children tasks within join scope of task \task{P}$'$ will join with it, and which then spawns task \task{T} copying \code{joined} from \task{P}$'$ to \task{T} (\task{P}$'$ can be \task{T} itself). Since tasks \task{T1} and \task{T2} have the highest LCA in the inheritance tree, and \task{P}$'$ is the common ancestor of \task{T1} and \task{T2}, \task{P}$'$ must be the ancestor of \task{T3} as well. Therefore, \task{T3} will be under join scope of \task{P}$'$ and so should be present in the \code{joined} of \task{P}$'$ after join call, and therefore to task \task{T} joined as well, like \task{T1} and \task{T2}.

    Sice \task{T3} would be present in \code{joined} of \task{T}, current access by \task{T} will have happens \emph{after} relationship with \task{T3}'s access and so would be non-racy. But we had assumed that \task{T3} access is racy with current \task{T} access, leading to a contradiction.
    This proves our lemma that any future access to the variable, which is racy with any of the existing parallel accesses, must be racy with at least one of the two access entries with the highest LCA.

\end{proof}

\begin{lemma}
    The inheritance vector clock (IVC) correctly computes the two tasks with the highest LCA among three parallel tasks.
\end{lemma}


\begin{figure}[t]
    \centering
    \begin{subfigure}[t]{0.45\linewidth}
        \centering
        {\includegraphics[width=\linewidth]{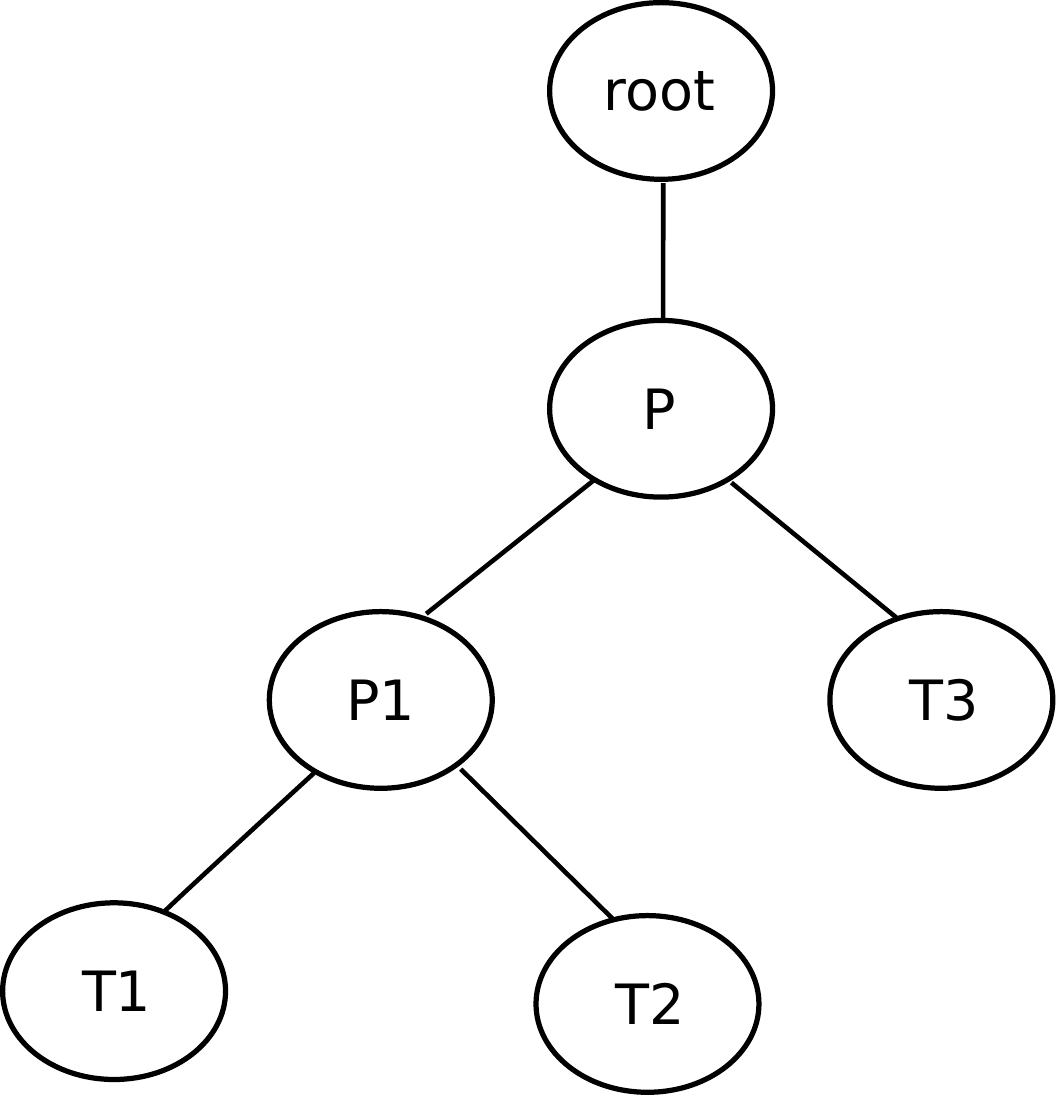}}
        \caption{}
        \label{fig:lca2a}
    \end{subfigure}%
    \hfill
    \begin{subfigure}[t]{0.45\linewidth}
        \centering
        {\includegraphics[width=\linewidth]{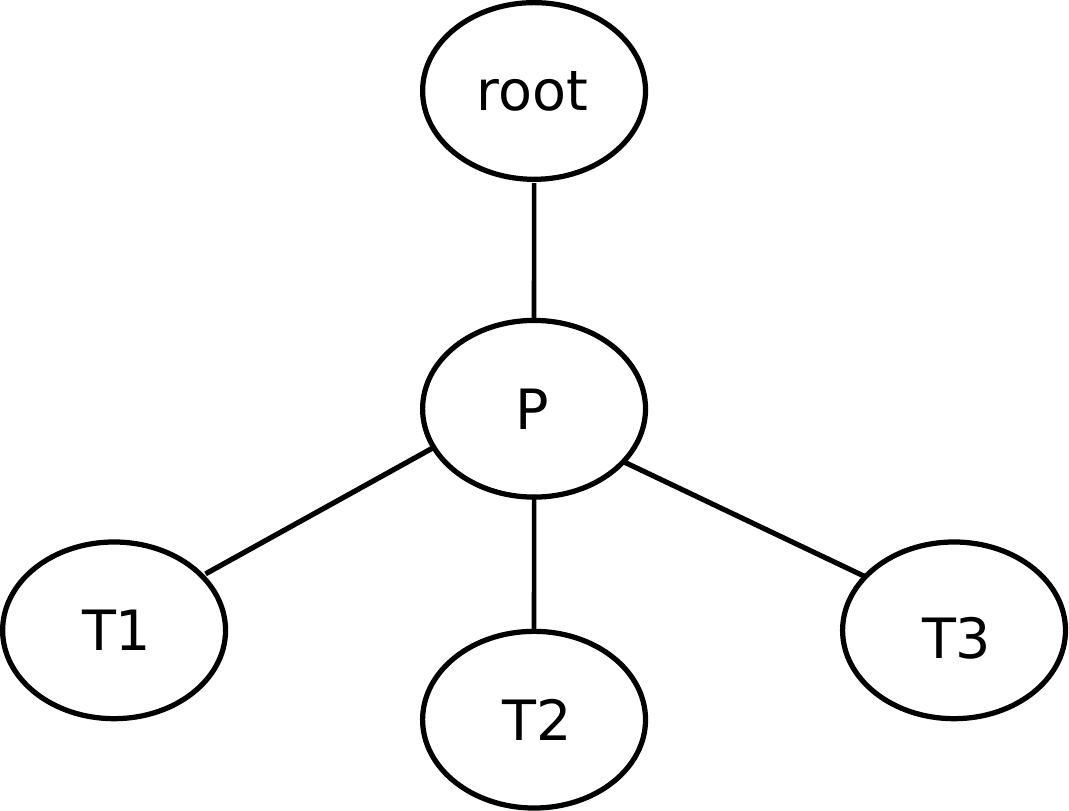}}
        \caption{}
        \label{fig:lca2b}
    \end{subfigure}
    \caption{Different cases possible during LCA computation.}
    \label{fig:lca-example}
\end{figure}

\begin{proof}

    Consider three parallel tasks \task{T1}, \task{T2}, and \task{T3} in an inheritance tree (Figure~\ref{fig:lca-example}). We need to consider the following two cases.

    \begin{description}
        \item [Case 1.] \lca{\task{T1}, \task{T2}}, \lca{\task{T2}, \task{T3}}, and \lca{\task{T1}, \task{T3}} are the same (Figure~\ref{fig:lca2a}).  Here, any pair of tasks chosen from the set \{T1,T2,T3\} will have same LCA, and so the claim holds.

        \item [Case 2.] Without loss of generality, let \lca{\task{T2}, \task{T3}} be strictly higher than \lca{\task{T1}, \task{T2}} (Figure~\ref{fig:lca2b}). Let \lca{\task{T1}, \task{T2}} = \task{P1} and \lca{\task{T2}, \task{T3}} = \task{P}. Then, \task{P} will be the ancestor of \task{P1}. By the structure of IVC, the IVC of a task contains the IVC of its ancestor as a prefix. So, the IVCs of tasks \task{T1} and \task{T2} will have the IVC of \task{P1} as a prefix, and the IVC of tasks \task{P1} and \task{T3} will have the IVC of \task{P} as a prefix. The IVCs of these tasks are as follows:
              \begin{itemize}
                  \item IVC(\task{P1}) = IVC(\task{P}) + \lbrack$a$\ldots$b$\rbrack
                  \item IVC(\task{T1}) = IVC(\task{P1}) + [$x_1$\ldots $x_n$] = IVC(\task{P}) + \lbrack$a$\ldots$b$\rbrack + [$x_1$\ldots $x_n$]
                  \item IVC(\task{T2}) = IVC(\task{P1}) + [$y_1$\ldots $y_m$] = IVC(\task{P)} + \lbrack$a$\ldots$b$\rbrack + [$y_1$\ldots $y_m$]
                  \item IVC(\task{T3}) = IVC(\task{P)} + [$z_1$\ldots $z_m$]
              \end{itemize}

              Here, $[x \ldots y]$ represents an array of clock values and IVC(\task{X}) represents inheritance vector clock of task \task{X}.
              Since, \lca{\task{T2}, \task{T3}}, \ie, \task{P} is strictly higher than \lca{\task{T1},\task{T2}} \ie, \task{P1},  IVC(\task{P}) cannot be equal to IVC(\task{P1}) and so [$a$ \ldots $b$] cannot be empty. Also, either [$z_1$\ldots $z_m$] is empty (\ie, \task{P} and \task{T3} are the same) or $z_1 \neq a$, otherwise \task{P} cannot be the LCA of \task{T2} and \task{T3}.

              While comparing the IVCs of tasks \task{T1}, \task{T2}, and \task{T3}, the first point of difference is after IVC(\task{P}). If [$z_1\ldots z_m$] is empty, then the IVC of \task{T3} will end first and so \fastpar will choose \task{T3} and any one of other two tasks (say \task{T2}). Otherwise, we will have $z_1$ in IVC(\task{T3}) and $a$ in IVC(\task{T1}) and IVC(\task{T2}). Since $z_1 \neq a$, the first point of difference has \task{T3}'s clock value as different. So, again \fastpar will choose \task{T3} and any one of other two (say \task{T2}). In  both cases, \fastpar chose \task{T3} and \task{T2} out of the three tasks. Since IVC(\task{T2},\task{T3}) = \task{P} is the highest LCA across all tasks in the set \{T1,T2, T3\}, so the claim holds.
    \end{description}

    The tasks \task{T1}, \task{T2}, and \task{T3} can be chosen randomly, and the above cases allow for any possible permutation.

\end{proof}

    }{}
\fi

\end{document}